\documentclass[11pt,journal,onecolumn]{IEEEtran}
\usepackage[utf8]{inputenc} 
\usepackage[T1]{fontenc}    
\usepackage[hidelinks]{hyperref}     

\usepackage{geometry} 
\usepackage{url}            
\usepackage{booktabs}       
\usepackage{amsfonts}       
\usepackage{nicefrac}       
\usepackage{microtype}      
\usepackage{yhmath} 
\usepackage{color}
\usepackage{epsfig}
\usepackage{graphicx}
\usepackage{amsmath}
\usepackage{amssymb}
\usepackage{amsthm}
\usepackage{bm}
\usepackage{algorithm}
\usepackage{algorithmicx}
\usepackage{algpseudocode}
\usepackage[numbers,sort&compress]{natbib}
\usepackage{epstopdf}
\usepackage{subfig}
\usepackage{graphicx}
\usepackage{physics}
\usepackage{setspace}
\usepackage{multirow}
\usepackage{authblk}
\setcitestyle{numbers,square,comma}
\graphicspath{{figures/}}
\allowdisplaybreaks

\newcommand{\bx}{\bm{x}}
\newcommand{\bs}{\bm{s}}

\newcommand{\bz}{\bm{z}}

\newcommand{\bX}{\bm{X}}
\newcommand{\bZ}{\bm{Z}}

\newcommand{\bL}{\bm{L}}
\newcommand{\bS}{\bm{S}}

\newcommand{\bM}{\bm{M}}
\newcommand{\bU}{\bm{U}}
\newcommand{\bV}{\bm{V}}

\newcommand{\bLa}{\bm{\Lambda}}
\newcommand{\bO}{\mathbf{0}}
\newcommand{\bP}{\bm{P}}
\newcommand{\bQ}{\bm{Q}}

\newcommand{\oM}{\widehat{\bm{M}}}
\newcommand{\oL}{\widehat{\bm{L}}}

\newcommand{\oZ}{\widehat{\bm{Z}}}

\newcommand{\ocH}{\widehat{\mathcal{H}}}

\newcommand{\hatL}{\widehat{\bL}}

\newcommand{\hatZ}{\widehat{\bZ}}
\newcommand{\hatcH}{\widehat{\cH}}
\newcommand{\cH}{\mathcal{H}}

\newcommand\numberthis{\addtocounter{equation}{1}\tag{\theequation}}

\DeclareMathOperator*{\minimize}{\mathrm{minimize}}
\DeclareMathOperator*{\subject}{\mathrm{subject~to~}}

\def\rank{\mathrm{rank}}
\def\supp{\mathrm{supp}}
\def\C{\mathbb{C}}
\def\la{\left\langle}
\def\ra{\right\rangle}
\def\lb{\left(}
\def\rb{\right)}

\def\ln{\left\|}
\def\rn{\right\|}

\def\aug{\widetilde}

\newtheorem{theorem}{Theorem}
\newtheorem{lemma}[theorem]{Lemma}

\newtheorem{definition}[theorem]{Definition}

\begin{document}

\title{Accelerated Structured Alternating Projections for Robust Spectrally Sparse Signal Recovery}

\author{HanQin Cai, Jian-Feng Cai, Tianming Wang, and Guojian Yin%
\thanks{H.Q. Cai is with the Department of Mathematics, University of California, Los Angeles, Los Angeles, California, USA (hqcai@math.ucla.edu).}
\thanks{J.-F. Cai is with the Department of Mathematics, Hong Kong University of Science and Technology, Clear Water Bay, Kowloon, Hong Kong SAR, China (jfcai@ust.hk).}
\thanks{T. Wang is with the School of Economic Mathematics, Southwestern University of Finance and Economics, Chengdu, Sichuan, China (corresponding author, wangtm@swufe.edu.cn).}
\thanks{G. Yin is with the Institute for Advanced Study, Shenzhen University, Shenzhen, Guangdong, China (yin@szu.edu.cn).}
}

\markboth{}{}

\maketitle

\begin{abstract}
Consider a spectrally sparse signal $\bx$ that consists of $r$ complex sinusoids with or without damping. We study the robust recovery problem for the spectrally sparse signal under the fully observed setting, which is about recovering $\bx$ and a sparse corruption vector $\bs$ from their sum $\bz=\bx+\bs$. In this paper, we exploit the low-rank property of the Hankel matrix formed by $\bx$, and formulate the problem as the robust recovery of a corrupted low-rank Hankel matrix. We develop a highly efficient non-convex algorithm, coined Accelerated Structured Alternating Projections (ASAP). The high computational efficiency and low space complexity of ASAP are achieved by fast computations involving structured matrices, and a subspace projection method for accelerated low-rank approximation. Theoretical recovery guarantee with a linear convergence rate has been established for ASAP, under some mild assumptions on $\bx$ and $\bs$. Empirical performance comparisons on both synthetic and real-world data confirm the advantages of ASAP, in terms of computational efficiency and robustness aspects.
\end{abstract}


\section{Introduction} \label{sec:intro}

In this paper, we study the corrupted spectrally sparse signal recovery problem under the fully observed setting.
Denote the imaginary unit by $\imath$. Let $x(t)$ be a continuous one-dimensional spectrally $r$-sparse signal; that is, $x(t)$  is a weighted superposition of $r$ complex sinusoids,
\begin{equation} \label{eq:cont_sss_definition}
    x(t)=\sum^r_{j=1} a_j e^{2\pi\imath f_j t - d_j t},
\end{equation}
where $a_j$, $f_j$ and $d_j$ represent non-zero complex amplitude, normalized frequency and damping factor of the $j$-th sinusoid, respectively. Let the column vector
\begin{equation*}
    \bx=[x(0),x(1),\cdots,x(n-1)]^T\in\C^n
\end{equation*}
denote the discrete samples of $x(t)$.

The spectrally sparse signal, as defined in \eqref{eq:cont_sss_definition}, appears in a wide range of applications including seismic imaging \citep{borcea2002imaging}, analog-to-digital conversion \citep{tropp2009beyond,cohen2018analog}, nuclear magnetic resonance (NMR) spectroscopy \citep{li1998high,holland2011fast,qu2015accelerated,nguyen2013denoising}, and fluorescence microscopy \citep{schermelleh2010guide}. Due to the malfunction of data acquisition sensors, the recorded signals are often corrupted by impulse noise, e.g., the baseline distortions in NMR \citep{xi2008baseline}. Traditional harmonic retrieval methods, such as Prony's method, ESPRIT \citep{roy1989esprit}, the matrix pencil method \citep{hua1990matrix}, and the finite rate of innovation approach \citep{vetterli2002sampling}, are often sensitive against such corruptions \citep{dragotti2007sampling}, and cannot be applied directly. Since the amount of impulse noise corruptions is often relatively small compared to the signal size, we can consider them as sparse corruptions. It is thus important to remove such sparse corruptions and recover the original signal accurately. 

For undamped signal with frequencies lying on the grids, the $\ell_1$ minimization based robust compressed sensing approach can successfully recover the original signal from corruptions \citep{li2013compressed}. However, the performance of this approach degrades if there is mismatch between the assumed on-the-grid frequencies and the true frequencies
\citep{chi2011sensitivity}. To handle the off-the-grid situation,
\citep{fernandez2017demixing} proposes an approach based on
total-variation norm minimization that can remove corruptions from the signal when the frequencies are sufficiently separated. Its theoretical results rely on the assumption that phases of the amplitude of the signal and the sparse components are uniformly distributed. The total-variation based minimization requires solving a SDP, which is computationally expensive in general.

Hankel formulation is another way to handle the off-the-grid frequencies, with an additional advantage to be able to model damping in the signals appeared in, e.g., NMR. In \citep{chen2014robust}, a convex method called Robust-EMaC is introduced with guaranteed recovery. It penalizes the nuclear norm to enforce the low-rank property of the Hankel matrix formed by the signal and the elementwise $\ell_1$ norm to promote the sparsity of Hankel matrix formed by the outliers. Its original formulation also requires solving a SDP. More recently, a non-convex algorithm named SAP has been proposed in \citep{zhang2018correction}. SAP alternatively projects the estimate of the Hankel matrix formed by the signal onto the set of low rank matrices, and the estimate of the outliers onto the set of sparse vectors. To avoid the potential negative effects from ill-conditioned matrices, SAP gradually increases the rank estimation from $1$ to $r$. Upon proper initialization, SAP has guaranteed linear convergence.
 
In this paper, we use the tangent space projection technique \citep{cai2019fast,cai2017accelerated} to further reduce the computational cost of SAP. The proposed method is coined Accelerated Structured Alternating Projections (ASAP), which costs significantly less flops per iteration and also has guaranteed linear convergence. Table~\ref{tab:comparison} compares Robust-EMaC, SAP, and ASAP in terms of tolerable corruption ratio, memory cost, and flops per iteration. Under the assumption that the locations of the outliers follow the Bernoulli distribution, Robust-EMaC is able to tolerate a (small) constant fraction of outliers. When employing a first-order solver, Robust-EMaC generally takes $\mathcal{O}(n^2)$ memory, and costs $\mathcal{O}(n^3)$ flops per iteration. SAP and ASAP do not impose restrictions on the outlier distribution, and their tolerable corruption ratios depend on some properties of the Hankel matrix formed by $\bx$. The definitions of $\mu$, $c_s$, and $\kappa$ are deferred to the next subsection. SAP takes $\mathcal{O}(nr)$ memory, and costs $\mathcal{O}(r^2n\log(n)$ flops per iteration, where the hidden constant in the complexity is large and depends on the relative gaps between the singular values of the corresponding Hankel matrix. Due to the tangent space projection technique, ASAP only costs $\mathcal{O}(r^2n+rn\log(n))$ flops per iteration, where the hidden constant is a fixed small number. In fact, ASAP can achieve $10\times$ speedup over SAP for large $n$ and small $r$.

\begin{table}[t]
\caption{Comparisons of Robust-EMaC  \citep{chen2014robust}, SAP \citep{zhang2018correction}, and ASAP [this paper]. Tolerance is the tolerable corruption ratio.
}
\label{tab:comparison}
\begin{center}
\begin{tabular}{| c | c | c |c|}
\hline 
& Tolerance & Memory & Flops Per Iteration \\
\hline
Robust-EMaC & $\mathcal{O}(1)$ & $\mathcal{O}(n^2)$ & $\mathcal{O}(n^3)$ \\
\hline
SAP & $\mathcal{O}(\frac{1}{\mu c_s r})$ & $\mathcal{O}(nr)$ & $\mathcal{O}(r^2n\log(n))$  \\
\hline
ASAP & $\mathcal{O}(\frac{1}{(\mu c_s r\kappa)^2})$ & $\mathcal{O}(nr)$ & $\mathcal{O}(r^2n+rn\log(n))$ \\
\hline
\end{tabular}
\end{center}
\end{table}

Our main contributions are two-fold. Firstly, we proposed a nonconvex method, ASAP, which has significantly improved computational efficiency over the state-of-art methods. Experiments also suggest that ASAP has similar tolerance for corruptions as SAP, despite the more stringent bound in theory. Secondly, we establish the local linear convergence of ASAP, and provide the initialization scheme. Although our analysis builds upon \citep{cai2017accelerated}, the extension to the Hankel case is by no means trivial. For one thing, the analysis in our case needs to deal with both the low-rank and the Hankel structure. We also improve the analysis in several aspects, such as removing the unnecessary trimming steps appeared in \citep{cai2017accelerated}.

\subsection{Formulation and Assumptions}

In this subsection, we describe the problem formulation in details, and present our assumptions for the theoretical analysis. Suppose we receive a corrupted signal $\bz\in\C^n$, which is the sum of an underlying spectrally sparse signal and some sparse corruptions, namely $\bz=\bx+\bs$. Our goal here is to recover $\bx$ and $\bs$ from $\bz$ simultaneously. This vector separating problem can be expressed as:
\begin{equation} \label{eq:vector model}
\begin{split}
    &\minimize_{\bx',\bs'} \|\bz-\bx'-\bs'\|_2^2  \\
    &\subject \bx' \textnormal{ is spectrally sparse and } \bs' \textnormal{ is sparse}.
\end{split}
\end{equation} 
That is, we are seeking a spectrally sparse vector $\bx'$ and sparse vector $\bs'$ 
such that their sum fits best to the given vector $\bz$.

One line of research works \citep{chen2014robust,cai2016robust,cai2019fast,cai2018spectral,xu2018sep} exploit the spectral sparsity of $\bx$ by the low-rankness of the Hankel matrix $\cH(\bx)$. Here, $\cH:\mathbb{C}^n\rightarrow \mathbb{C}^{n_1\times n_2}$ is a mapping from a complex vector to a complex Hankel matrix, where $n=n_1+n_2-1$. For $\bm{x}=[x_0;x_1;\cdots;x_{n-1}]\in\mathbb{C}^{n}$, we have%
\begin{align} \label{eq:Hankel definition}
\cH(\bm{x}) = 
\begin{bmatrix}
x_0 & x_1  & \cdots & x_{n_2-1}\\
x_1 & x_2  & \cdots & x_{n_2}  \\
\vdots &  \vdots & \cdots  &\vdots\\
x_{n_1-1} & x_{n_1}  & \cdots   &x_{n-1}
\end{bmatrix}\in\C^{n_1\times n_2}.
\end{align}
A good choice of the matrix size is $n_1\approx n_2$, i.e., we want to construct a nearly squared Hankel matrix \citep{chen2014robust}. Without loss of generality, throughout this paper, we use $n_1=n_2=(n+1)/2$ if $n$ is odd, and $n_1=n_2-1=n/2$ if otherwise.

For the spectrally $r$-sparse signal $\bx\in\mathbb{C}^n$, $\cH(\bx)$ has a Vandermonde decomposition in the form of
$$
\cH(\bx)=\bm{E}_L \bm{D} \bm{E}_R^T,
$$
where $\bm{D}=\textnormal{diag}{([a_1,\cdots,a_r])}$, $[\bm{E}_L]_{i_1,j}=\omega^{i_1-1}_j$, $[\bm{E}_R]_{i_2,j}=\omega^{i_2-1}_j$, and $\omega_j=e^{2\pi\imath f_j - d_j}$ for $i_1\in\{1,\cdots,n_1\}$, $i_2=\{1,\cdots,n_2\}$ and $j\in\{1,\cdots,r\}$. If the pairs $(f_j,d_j)$ are different, the left and right matrices in this Vandermonde decomposition are both of full rank. Thus, $\cH(\bx)$ is rank $r$ provided all the complex amplitudes $\{a_j\}$ are non-zero. It is worth mentioning that the low-rank Hankel matrix also appears in many other applications, such as impulse noise removal \citep{Jin2018}, magnetic resonance imaging (MRI) \citep{haldar2014low,jin2016general}, dynamical system identification \citep{shah2012linear,fazel2013hankel}, and autoregressive moving average \citep{kung1983state,akaike1998markovian}. Sparse corruptions also appear in these applications, e.g., the Herringbone artifact in MRI \citep{jin2017mri}.

With the properties of Hankel operator $\cH$, we can reformulate the corrupted spectrally sparse signal recovery problem \eqref{eq:vector model} into the robust recovery of a corrupted low-rank Hankel matrix: 
\begin{equation} \label{eq:matrix model 1}
\begin{split}
        &\minimize_{\bx', \bs'} \|\cH(\bz)-\cH(\bx'+\bs')\|_F^2         \\
        &\subject \rank(\cH(\bx')) = r \textnormal{ and } \|\bs'\|_0 \leq \alpha n.
\end{split}
\end{equation}
Here, $\|\cdot\|_F$ denotes the Frobenius norm of matrices, 
$\|\cdot\|_0$ counts the number of non-zero entries, and $\alpha$ is the expected sparsity level of the underlying corruptions. Note that $\cH$ is an injective mapping. Hence, the reconstruction of $\bx$ and $\cH(\bx)$ are equivalent.

Without any further assumptions, the optimization problem \eqref{eq:matrix model 1} is clearly ill-posed. Inspired by robust principal component analysis (RPCA) \citep{wright2009robust,netrapalli2014non,NIPS2016_b5f1e8fb,cai2017accelerated}, we assume that $\cH(\bx)$ is not too sparse, and $\bs$ is not too dense. These assumptions are formalized in \nameref{assume:Inco} and \nameref{assume:Sparse}, respectively.

\noindent\paragraph*{A1}\label{assume:Inco} 
\textit{The Hankel matrix $\cH(\bx)\in \mathbb{C}^{n_1\times n_2}$ corresponding to the underlying spectrally $r$-sparse signal $\bx\in\C^{n}$ is ${\mu}$-incoherent, that is
\begin{equation*}
\|\bm{U}\|_{2,\infty}\leq \sqrt{\frac{\mu c_s r}{n}} \quad\textnormal{and}\quad \| \bm{V}\|_{2,\infty}\leq \sqrt{\frac{\mu c_s r}{n}},
\end{equation*}
where $\|\cdot\|_{2,\infty}$ is the maximum of the $l_2$ norms of the rows, $c_s:=\max\{n/n_1,n/n_2\}$, and $\bm{U}\bm{\Sigma} \bm{V}^*$ is the singular value decomposition (SVD) of $\cH(\bx)$.}

Theoretically, in the undamped case, \nameref{assume:Inco} holds if the minimum wrap-around distances between the frequencies $\{f_j\}$ are greater than $2/n$ \citep[Theorem 2]{liao2016music}. Empirically, we find that \nameref{assume:Inco} is commonly satisfied for randomly generated spectrally sparse signals and real-world examples, with or without damping factors. 

\paragraph*{A2}\label{assume:Sparse}
\textit{The corruption vector $\bs\in\C^n$ is $\alpha $-sparse, i.e., $\bs$ has at most $\alpha n$ non-zero entries. In this paper, we assume $\alpha\lesssim\mathcal{O}\left(\frac{1}{(\mu c_s r\kappa)^2}\right),$
where $\kappa$ is the condition number of $\cH(\bx)$. 
}

The notation `$\lesssim$' means that there exists an absolute constant $c>0$ such that $\alpha$ is bounded by $c$ times the right hand side. Let $\sigma_i^x$ denote the $i$-th singular value of $\cH(\bx)$, $\kappa=\sigma_{1}^x/\sigma_{r}^x$. If the minimum warp-around distances are greater than $2/n$, then the condition number $\kappa$ of the corresponding Hankel matrix depends on the largest and smallest magnitudes of the complex amplitudes \citep[Remark~1]{cai2019fast}.  
Essentially, \nameref{assume:Sparse} states that there cannot be too many corruptions in the signal. 
Note that we make no assumption on the distribution of the corruptions since $\cH(\bs)$ preserves the sparsity of $\bs$. Indeed, if $\bs\in\C^n$ is $\alpha$-sparse, there are no more than $\alpha n$ non-zero entries in each row and column of $\cH(\bs)$.

\subsection{Notations and Paper Organization}

In this paper, we denote column vectors by bold lowercase letters (e.g., $\bm{v}$), matrices by bold capital letters (e.g., $\bm{M}$), and operators by calligraphic letters (e.g., $\mathcal{P}$). For any vector $\bm{v}$, $\|\bm{v}\|_2$ and $\|\bm{v}\|_\infty$ denotes the $\ell_2$ norm and $\ell_\infty$ norm of $\bm{v}$, respectively. For any matrix $\bm{M}$, $\|\bm{M}\|_{2,\infty}$ denotes the maximum of the $l_2$ norms of the rows, $[\bm{M}]_{i,j}$ denotes its $(i,j)$-th entry, $\|\bm{M}\|_\infty=\max_{i,j} |[\bm{M}]_{i,j} |$ denotes the maximum magnitude among its entries, $\sigma_i(\bm{M})$ denotes its $i$-th singular value, $\|\bm{M}\|_2=\sigma_1(\bm{M})$ denotes its spectral norm, $\|\bM\|_F=\sqrt{\sum_i \sigma_i^2(\bm{M})}$ denotes its Frobenius norm, and $\|\bM\|_\ast=\sum_i \sigma_i(\bM)$ denotes its nuclear norm. Furthermore, $\la\cdot,\cdot\ra$, $\overline{(\cdot)}$, $(\cdot)^T$, and $(\cdot)^\ast$ denote the inner product, conjugate, transpose, and conjugate transpose, respectively.

In particular, we use $\bm{e}_i$ to denote the $i$-th canonical basis vector, $\bm{I}$ to denote the identity matrix, and $\mathcal{I}$ to denote the identity operator. Throughout this paper, $\bm{L}=\cH(\bx)$ denotes the underlying rank $r$ Hankel matrix. At the $k$-th iteration, the estimates of $\bL$, $\bx$ and $\bs$ are denoted by $\bL_k$, $\bx_k$ and $\bs_k$, respectively. 

We organize the rest of the paper as follows. Sections~\ref{subsec:main_algo} and \ref{subsec:init} present the proposed main algorithm and the corresponding initialization scheme, respectively. The theoretical results of the proposed algorithm are presented in Section~\ref{subsec:recovery}, followed by the discussion on extending the algorithm and theoretical results to the multi-dimensional cases. Section~\ref{sec:numberical} contains extensive numerical experiments of our algorithm, on both synthetic and real-world datasets. All the mathematical proofs of our theoretical results are presented in Section~\ref{sec:proof}. The paper is concluded with some future directions in Section~\ref{sec:conclusion}.
\section{Algorithms} \label{sec:Algorithm}

It is clear that robust low-rank Hankel matrix recovery problem \eqref{eq:matrix model 1} can be viewed as a RPCA problem, and we can solve it with any off-the-shelf RPCA algorithm. However, without taking advantage of the Hankel structure, we cannot achieve the optimal computational efficiency and robustness. Inspired by an accelerated alternating projections (AccAltProj) algorithm for RPCA introduced in \citep{cai2017accelerated}, we present an algorithm for problem \eqref{eq:matrix model 1}, dubbed Accelerated Structured Alternating Projections (ASAP). While enjoying theoretical guaranteed recovery, the proposed algorithm has improved computational efficiency compared to state-of-art methods.

The proposed algorithm proceeds in two phases. In the first phase, we initialize the algorithm by one step alternating projections. In the second phase, we project $\bz-\bx_k$ onto the space of sparse vectors to get the update $\bs_{k+1}$, and then compute the update $\bx_{k+1}$ from the accelerated rank $r$ approximation of $\cH(\bz-\bs_{k+1})$.

\subsection{Main Algorithm}  \label{subsec:main_algo}

Firstly, we will discuss the second phase - the main algorithm, which is summarized in Algorithm~\ref{Algo:Algo1}. 

\begin{algorithm}[t]
\caption{\textbf{A}ccelerated \textbf{S}tructured \textbf{A}lternating \textbf{P}rojections} \label{Algo:Algo1}
\begin{algorithmic}[1]
\State \textbf{Input:} $\bz=\bx+\bs$: observed corrupted signal; $r$: model order; $\varepsilon$: target precision level; $\beta$: thresholding parameter; $\gamma$: thresholding decay parameter.
\State \textbf{Initialization} and set $k=0$
\While{ \texttt{$\|\bz-\bx_{k}-\bs_{k}\|_2/\|\bz\|_2 \geq \varepsilon$} }
    \State $\zeta_{k+1}= \beta~\gamma^{k}\sigma_1(\bL_{k}) $
    \State $\bs_{k+1}=\mathcal{T}_{\zeta_{k+1}}(\bz-\bx_{k})$
    \State $\bL_{k+1}=\mathcal{D}_r\mathcal{P}_{T_{k}}\cH(\bz-\bs_{k+1})$
    \State $ \bx_{k+1} = \cH^{\dagger}(\bm{L}_{k+1}) $
    \State $k=k+1$
\EndWhile
\State \textbf{Output:}  $\bx_{k}$
\end{algorithmic}
\end{algorithm}

Define the hard thresholding operator $\mathcal{T}_\zeta$ as
\begin{equation*}
    [\mathcal{T}_{\zeta}\bm{v}]_{t} =
    \begin{cases}
    [\bm{v}]_{t} & \quad|[\bm{v}]_{t}| >\zeta,\\
    0  & \quad\mbox{otherwise.}
    \end{cases}
\end{equation*} ASAP starts with updating the estimate of $\bs$ by projecting $\bz-\bx_{k}$ onto the space of sparse vectors:
\begin{align*}
    \bm{s}_{k+1}=\mathcal{T}_{\zeta_{k+1}}(\bz-\bx_{k}).
\end{align*}
The key to successful isolation of corruptions is the choice of proper thresholding value. At the $(k+1)$-th iteration, ASAP selects the hard thresholding value as 
\begin{equation*}
    \zeta_{k+1}= \beta~\gamma^{k}\sigma_1(\bL_{k}),
\end{equation*}
where $\beta$ is a positive tuning parameter, $\gamma\in(0,1)$ is a decay parameter, and $\sigma_1(\bL_{k})$ has been computed in the previous iteration (see \eqref{eq:step:get new SVD} later). Thus, the computational cost of $\zeta_{k+1}$ is negligible, and the total cost of updating $\bs$ is $\mathcal{O}(n)$ flops.

Next, we will update the estimate of $\bL$, i.e., $\cH(\bx)$. We consider a low-dimensional subspace $T_k$ formed by the direct sum of the column and row spaces of $\bL_k$, i.e.,
\begin{equation}  \label{eq:tangent space tilde k}
T_k=\left\{\bU_k\bm{A}^*+\bm{B}\bV_k^* ~|~\bm{A}\in\mathbb{C}^{n_2\times r}, \bm{B}\in\mathbb{C}^{n_1\times r} \right\},
\end{equation}
where $\bL_k=\bU_k{\bm{\Sigma}}_k\bV_k^*$ is its SVD. The subspace $T_k$ can be viewed as the tangent space of the rank $r$ matrix manifold at $\bL_k$ \citep{vandereycken2013low}, and it has been widely studied in the low-rank matrices related recovery problems \citep{absil2009optimization,recht2011simpler,ngo2012scaled,mishra2014fixed,wei2016recovery}. Moreover, for any $\bM\in\C^{n_1\times n_2}$, the projection of $\bM$ onto the low-dimensional subspace $T_k$ can be computed by
\begin{equation}  \label{eq:projection onto tangent space tilde k}
\mathcal{P}_{T_k}\bM=\bU_k\bU_k^*\bM+\bM\bV_k\bV_k^*-\bU_k\bU_k^*\bM\bV_k\bV_k^*.
\end{equation}

To get new estimate $\bL_{k+1}$, we first project Hankel matrix $\cH(\bz-\bs_k)$ onto the low-dimensional subspace $T_k$, and then project onto the set of rank $r$ matrices. That is
\begin{equation}  \label{eq:step: update L}
    \bm{L}_{k+1}=\mathcal{D}_r\mathcal{P}_{T_{k}}\cH(\bz-\bs_{k}),
\end{equation}
where $\mathcal{D}_r$ computes the nearest rank $r$ approximation via truncated SVD. 
Although there is a SVD in this step, we can compute it efficiently by using the properties of the low-dimensional subspace $T_k$ \citep{vandereycken2013low,wei2016recovery,cai2017accelerated,cai2019fast}. Denote $\bm{H}_k:=\cH(\bz-\bs_k)$. Let $(\bm{I}-\bV_k\bV_k^*)\bm{H}_k^*\bU_k=\bm{Q}_1\bm{R}_1$ and $(\bm{I}-\bU_k\bU_k^*)\bm{H}_k\bV_k=\bm{Q}_2\bm{R}_2$ be the QR-decompositions. Since $\bV_k \perp \bm{Q}_1$ and $\bU_k \perp \bm{Q}_2$,
\begin{equation*}
\mathcal{P}_{T_{k}} \bm{H}_k 
		= [\bU_k~\bm{Q}_2] \bm{M}_k [\bV_k~\bm{Q}_1]^*,
\end{equation*}
where 
$\bM_k:=\begin{bmatrix}  
\bU_k^*\bm{H}_k\bV_k & \bm{R}_1^*  \\ 
\bm{R}_2 & \bO   
\end{bmatrix}$ 
is a $2r\times 2r$ matrix. Let $\bm{U}_{M}\bm{\Sigma}_{M}\bm{V}_{M}^*$ be the rank $r$ truncated SVD of $\bM_k$. Since both $[\bU_k~\bm{Q}_2]$ and $[\bV_k~\bm{Q}_1]$ are orthonormal matrices, we can obtain the SVD of $\bL_{k+1}$ as
\begin{align}  \label{eq:step:get new SVD}
&~\left([\bU_k~\bm{Q}_2]\bm{U}_{M}\right) \bm{\Sigma}_{M} \left([\bV_k~\bm{Q}_1]\bm{V}_{M}\right)^*.
\end{align}
Altogether, the computation of \eqref{eq:step: update L} consists of the multiplication between a $n_1 \times n_2$ Hankel matrix and a $n_2 \times r$ matrix, the multiplication between a $n_2 \times n_1$ Hankel matrix and a $n_1 \times r$ matrix, two QR-decompositions of sizes $n_1 \times r$ and $n_2\times r$, and a truncated SVD of a $2r\times 2r$ matrix. The Hankel matrix-vector multiplication can be computed efficiently without forming the Hankel matrix explicitly via FFT, which costs only $\mathcal{O}(n\log(n))$ flops \citep{lu2015fast}. Hence, updating $\bL$ estimate requires $\mathcal{O}(r^2n+rn\log(n)+r^3)$ flops where the hidden constant is a fixed small number.

Finally, we update the estimate of $\bx$ from $\bL_{k+1}$:
\begin{equation} \label{eq:update x}
    \bx_{k+1}= \cH^\dagger(\bL_{k+1}),
\end{equation}
where $\cH^{\dagger}:\C^{n_1 \times n_2}\rightarrow \C^n$ denote the left inverse of $\cH$, i.e., $\cH^\dagger\cH=\mathcal{I}$. For any matrix $\bM\in \C^{n_1 \times n_2}$,
\begin{equation} \label{eq:Hankel inverse mapping definition}
    [\cH^{\dagger}(\bM)]_t=\frac{1}{\rho_t}\sum_{a+b=t}[\bM]_{a,b}
\end{equation}
for $0\leq t \leq n-1$, where $\rho_t$ is the number of entries on the $t$-th anti-diagonal of $\bM$. Furthermore,
\begin{equation*}  
    \cH^\dagger(\bL_{k+1})= \sum_{j=1}^r [\bm{\Sigma}_{k+1}]_{j,j} \cH^\dagger\left([\bU_{k+1}]_{:,j} \left([\bV_{k+1}]_{:,j}\right)^*\right),
\end{equation*}
where 
$$[\cH^\dagger([\bU_{k+1}]_{:,j} ([\bV_{k+1}]_{:,j})^*)]_t=\frac{1}{\rho_t}\sum_{a+b=t} [\bU_{k+1}]_{a,j} [\overline{\bV}_{k+1}]_{b,j},$$ 
and it can be computed via fast convolution. Thus, the computational costs of \eqref{eq:update x} is $\mathcal{O}(rn\log(n))$.

Although the last three steps of ASAP involve $n_1\times n_2$ matrices, the entire process does not require forming these matrices explicitly. We only need to store the corresponding vector for the Hankel matrix and the SVD components for the rank $r$ matrix. Therefore, the space complexity of ASAP is $\mathcal{O}(rn)$ instead of $\mathcal{O}(n^2)$.

\subsection{Initialization}  \label{subsec:init}

To apply the proposed acceleration method for rank $r$ matrix projection, i.e., \eqref{eq:step: update L}, we need to form the low-dimensional subspace $T_0$ by the singular vectors of a reasonably estimated $\bL_0$. We propose an initialization method based on one-step alternating projections, which is summarized in Algorithm~\ref{Algo:Init1}. The primary difference between the initialization scheme and the main algorithm is a truncated SVD (without acceleration) needs to be computed instead. However, the matrix-vector multiplications involved in the truncated SVD of a Hankel matrix can be computed via FFT without forming the Hankel matrix explicitly; that is, the computational complexity for updating $\bL_0$ is $\mathcal{O}(rn\log(n))$ with a large hidden constant depending on the gaps between the singular values, and the space complexity remains $\mathcal{O}(rn)$ as we only need to store the singular vectors of $\bL_0$. It is worth mentioning that when we make the initial guess of $\bs$, a thresholding parameter $\beta_{init}$ is used to offset the spectral perturbation caused by corruptions, which may be turned differently than $\beta$ in Algorithm~\ref{Algo:Algo1}.

\begin{algorithm}[t]
\caption{Initialization}\label{Algo:Init1}
\begin{algorithmic}[1]
\State \textbf{Input:} $\bz=\bx+\bs$: observed corrupted signal; $r$: model order; $\beta_{init}$: thresholding parameter for initialization.
\State $\zeta_{0} = \beta_{init} ~ \sigma_1(\cH(\bz))$
\State $\bs_{0}=\mathcal{T}_{\zeta_{0}}(\bz)$
\State $\bL_0=\mathcal{D}_r\cH(\bz-\bs_{0})$
\State $\bx_{0}=\cH^{\dagger}(\bL_{0})$
\State \textbf{Output:} $\bL_0$, $\bx_0$.
\end{algorithmic}
\end{algorithm}

In summary, the overall space complexity of the new algorithm is $\mathcal{O}(rn)$, the computational cost per iteration is $\mathcal{O}(r^2n+rn\log(n))$ flops with fixed small constant in the front, and additional $\mathcal{O}(rn\log(n))$ flops with relatively large hidden constant at initialization. The computational and space efficiency of ASAP is then established when $r$ is small and $n$ is large, which is later confirmed again by the empirical experiments in Section~\ref{sec:numberical}. For reader's convenience, a sample MATLAB implementation of ASAP is provided at
\begin{center}
    \textit{\url{https://github.com/caesarcai/ASAP-Hankel}}.
\end{center}
\subsection{Recovery Guarantee}  \label{subsec:recovery}

The theoretical results of the proposed algorithm are presented in this section while the proofs are presented later in Section~\ref{sec:proof}. We begin with the local convergence guarantee of Algorithm~\ref{Algo:Algo1}.

\begin{theorem}[Local Convergence]  \label{thm:local convergence}
Let $\bx, \bs\in\C^{n}$ satisfy Assumptions \nameref{assume:Inco} and \nameref{assume:Sparse}, respectively. Suppose $\beta=\frac{\mu c_s r}{2\kappa n }$ and $\gamma\in\left(\frac{1}{\sqrt{12}},1\right)$. Denote $\tau=4\alpha\mu c_s r\kappa$. If the initial guess obeys the following conditions:
\[
\|\bL-\bL_0\|_2 \leq 2\tau \sigma_r^x,\quad
\|\bx-\bx_0\|_\infty \leq \frac{\tau-2\tau^2}{8\alpha \kappa n } \sigma_r^x,
\]
and $\bL_0$ is $4\mu \kappa^2$-incoherent,
then iterates of Algorithm~\ref{Algo:Algo1} satisfy 
\[
\|\bL-\bL_k\|_2 \leq 2\tau \gamma^k\sigma_r^x,\quad
\|\bx-\bx_k\|_\infty \leq \frac{\tau-2\tau^2}{8\alpha \kappa n } \gamma^k\sigma_r^x,
\]
and $\bL_k$ is also $4\mu \kappa^2$-incoherent.
\end{theorem}

As Theorem~\ref{thm:local convergence} requires a sufficiently close initialization for the local convergence, the following theorem provides the conditions such that the outputs of Algorithm~\ref{Algo:Init1} is inside the desired basin of attraction.

\begin{theorem}[Sufficient Initialization] \label{thm:initialization bound} 
Let $\bx, \bs\in\C^{n}$ satisfy Assumptions~\nameref{assume:Inco} and \nameref{assume:Sparse}, respectively. Suppose $\frac{\mu c_sr\sigma_1^x}{n\sigma_1(\cH(\bz))}\leq\beta_{init}\leq\frac{3\mu c_s r\sigma_1^x}{n\sigma_1(\cH(\bz))}$. Denote $\tau=4\alpha\mu c_s r\kappa$. Then the outputs of  Algorithm~\ref{Algo:Init1} satisfy
\[
\|\bL-\bL_0\|_2 \leq 2\tau \sigma_r^x,\quad
\|\bx-\bx_0\|_\infty \leq \frac{\tau-2\tau^2}{8\alpha \kappa n } \sigma_r^x,
\]
and $\bL_0$ is $4\mu \kappa^2$-incoherent.
\end{theorem}

By the definition of $\cH^{\dagger}$ in \eqref{eq:Hankel inverse mapping definition}, it is clear that $\|\bx-\bx_k\|_2=\|\cH^{\dagger}(\bL-\bL_k)\|_2\leq\|\bL-\bL_k\|_F\leq\sqrt{2r}\|\bL-\bL_k\|_2$. Combining Theorem~\ref{thm:local convergence} with Theorem~\ref{thm:initialization bound}, it establishes the linear convergence of $\bx_k$ to $\bx$ for ASAP.

\textbf{Remark 1:} The proof of Theorem~\ref{thm:local convergence} can be done with $\beta=\frac{\mu c_sr}{2n}$, which is similar to the parameter setting in \citep{zhang2018correction}, but it will reduce the toleration bound of $\alpha$ to the order of $1/\kappa^3$. By setting $\beta=\frac{\mu c_sr}{2\kappa n}$, we manage to improve the dependence of $\alpha$ on the condition number. In practice, $\kappa$ and $\mu$ may be estimated from initialization or prior knowledge. In either setting for $\beta$, the theoretical corruption level that can be tolerated by ASAP is worse than the optimal guarantee $\mathcal{O}(1/\mu c_s r)$ obtained in \citep{zhang2018correction}. Here, the looseness of an order in $\mu c_s r$ may be improved if we further tune the size of the neighborhood for the local convergence analysis, i.e., striking a better balance between the requirements of $\alpha$ for the local convergence analysis, and the initialization. Also, the appearance of condition number $\kappa$ in our requirement is due to the fixed rank setting. Nonetheless, empirical results show that ASAP can indeed tolerate more corruptions in practice, thus the theoretical requirement is highly pessimistic. 

\textbf{Remark 2:} The choice of $\beta_{init}$ in Theorem~2 seems to require a lot prior knowledge about the ground truth $\bx$, but what we really need is the bound of $\|\bx\|_{\infty}$. In Part I of the proof for initialization, $\beta_{init}$ is chosen such that the initial threshold value $\zeta_0\geq\|\bx\|_{\infty}$, ensuring that $\text{supp}(\bs_0)\subset\text{supp}(\bs)$. Thus one can tune $\zeta_0$ directly. In real applications, it is very likely to have prior knowledge about the range of a clean signal. Some overestimate of $\|\bx\|_{\infty}$ can also be tolerated in our analysis.

\subsection{Extension to Multi-Dimensional Signals} \label{sec:extendion to mulit dimension}
For ease of presentation, we focus on the one-dimension signal case in this paper; however, our algorithm and the corresponding theoretical results can be easily extended to the multi-dimensional cases. For any $N$-dimensional signal, we can define a Hankel operator $\cH_N$ that maps a $N$-dimensional tensor to a $N$-level Hankel matrix. The definition of $\cH_N$ and the corresponding left inverse $\cH_N^\dagger$ can be found in, e.g., Section~2.4 of \citep{cai2019fast}. We emphasize that all the key properties of $\cH$ and $\cH^\dagger$ are well retained with the properly defined $\cH_N$ and $\cH_N^\dagger$. For instance, given a tensor $\bX$ corresponding to the samples of a $N$-dimensional spectrally $r$-sparse signal, one can easily verify that $\cH_N(\bX)$ is rank $r$. The incoherence assumption \nameref{assume:Inco} is still guaranteed for undamped signals with sufficient wrap-around distances between the frequencies. Moreover, for assumption \nameref{assume:Sparse}, $\cH_N(\bS)$ preserves the sparsity level of $\bS$. Following the proofs in Section~\ref{sec:proof}, one can directly extend our theoretical results to the multi-dimensional cases.
\section{Numerical Experiments} \label{sec:numberical}

In this section, we conduct numerical experiments to evaluate the empirical performance of ASAP. The experiments are executed from MATLAB on a Windows 10 laptop with Intel i7-8750H CPU and 32GB of RAM. To match SAP for a fair comparison, we employ the turning parameters $\beta_{init}=\frac{2\mu c_s r\sigma_1^x}{n \sigma_1(\cH(\bz))}$ and $\beta=\frac{\mu c_s r}{2n}$ in our experiments, thus we need the estimates of $\mu$ and $\sigma_1^x$. If prior information about $\|\bm{x}\|_{\infty}$ is available, we do not need to estimate $\beta_{init}$. Empirically, we find one step of Cadzow \citep{cadzow1988signal}, which costs $\mathcal{O}(rn\log(n))$ flops, provides good estimates of those values. Therefore such a routine is included in ASAP. The truncated SVD in the initialization is computed using the PROPACK package \citep{larsen2004propack}. The relative error at the $k$-th iteration is defined as $err_{k}=\|\bm{z}-\bm{x}_k-\bm{s}_k\|_2/\|\bm{z}\|_2$. ASAP is terminated when either $err_{k}$ is below a threshold $tol$, or the iteration number is greater than $100$. 

\subsection{Empirical Phase Transition}

We evaluate the recovery ability of ASAP and compare it with TV \citep{fernandez2017demixing}, Robust-EMaC \citep{chen2014robust}, and SAP \citep{zhang2018correction}. The code of \citep{fernandez2017demixing} can be found in its supplementary material. Robust-EMaC is implemented using CVX \citep{grant2008cvx} with default parameters. We implement SAP ourselves. The test spectrally sparse signals of length $n$ with $r$ frequency components are formed in the following way: each frequency $f_j$ is randomly generated from $[0,1)$, and the argument of each complex coefficient $a_j$ is uniformly sampled from $[0,2\pi)$ while the amplitude is selected to be $1+10^{0.5c_k}$ with $c_k$ being uniformly distributed over $[0,1]$. We test two different settings for the frequencies: a) no separation condition is imposed  on $\{f_k\}_{k=1}^r$, and b) the wrap-around distances between each pair of the randomly drawn frequencies are guaranteed to be greater than $1.5/n$. The locations of the corruptions are chosen uniformly, while the real and imaginary parts of the corruptions are drawn i.i.d. from the uniform distribution over the interval $[-c~\mathbb{E}(|\Re([\bx]_i)|),c~\mathbb{E}(|\Re([\bx]_i)|)]$ and $[-c~\mathbb{E}(|\Im([\bx]_i)|),c~\mathbb{E}(|\Im([\bx]_i)|)]$ for some constant $c>0$, respectively. For a given triple $(n,r,\alpha)$, $50$ random tests are conducted. We consider an algorithm to have successfully reconstructed a test signal if the recovered signal $\bx_{\mathrm{rec}}$ satisfies $\|\bx_{\mathrm{rec}}-\bx\|_2/\|\bx\|_2\leq 10^{-3}$. The tests are conducted with $n=125$ and $c=1$. An important parameter for both SAP and our algorithm is the target convergence rate $\gamma$. For easier problems, a smaller $\gamma$ can be chosen for computation efficiency. Since we would like to test the limits of the algorithms' recovery abilities, $\gamma$ is set to $0.95$ for both algorithms. We also tune the sparsity penalty parameter for TV and Robust-EMaC, and report the best performances among the chosen parameters. 

\begin{figure}[t]
\centering
\subfloat{\includegraphics[width=.22\linewidth]{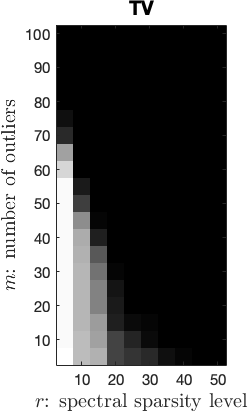}} \hfill
\subfloat{\includegraphics[width=.22\linewidth]{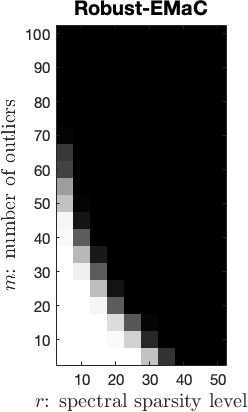}} \hfill
\subfloat{\includegraphics[width=.22\linewidth]{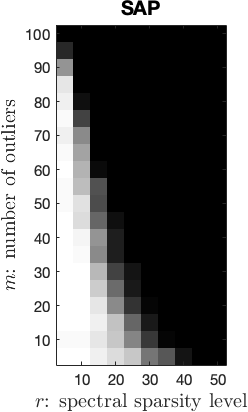}} \hfill
\subfloat{\includegraphics[width=.22\linewidth]{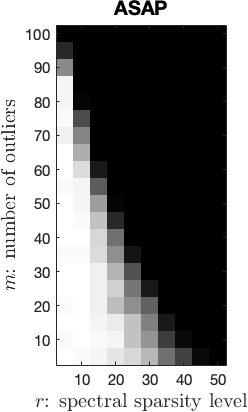}}\\
\subfloat{\includegraphics[width=.22\linewidth]{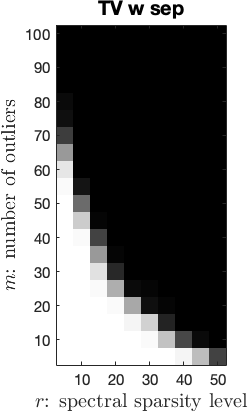}} \hfill
\subfloat{\includegraphics[width=.22\linewidth]{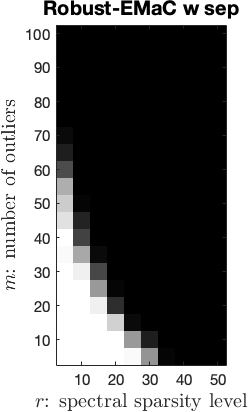}} \hfill
\subfloat{\includegraphics[width=.22\linewidth]{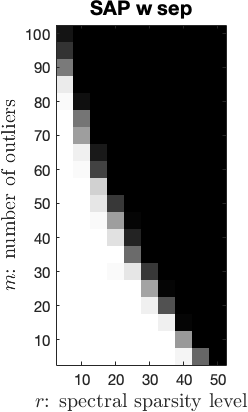}} \hfill
\subfloat{\includegraphics[width=.22\linewidth]{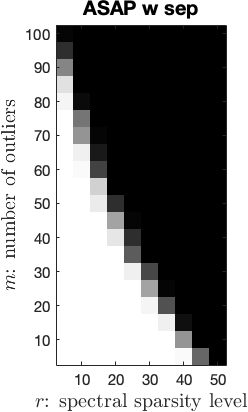}}
\caption{Phase transition comparisons: $x$-axis is spectral sparsity level $r$ and $y$-axis is number of corruptions $m$. \textbf{Top:} no restriction on frequencies of test signals. \textbf{Bottom:} wrap-around distances between frequencies are at least $1.5/n$.} \label{fig:PT}
\end{figure}

In Figure~\ref{fig:PT}, we can observe that for signals with sufficient frequency separations, TV is better than Robust-EMaC, but is not as good as the nonconvex methods. The performance of TV degrades severely when the signal frequencies are not well separated. The two non-convex methods have better performances than the convex method Robust-EMaC, and ASAP is more favorable compared to SAP for harder problems when the frequencies are less separated.   

\subsection{Computational Efficiency}

Next, we compare ASAP with SAP in terms of computational efficiency. For fair comparisons, we modify SAP such that it does not need to gradually increase the rank. The experiments are conducted on 2D spectrally sparse signals, whose definition can be found in e.g., Section 2.4 of \citep{cai2019fast}. The tested signals are square matrices with various sizes, generated similarly as in the 1D case without frequency separation. Such sizes are prohibitive for Robust-EMaC, even with a first-order solver. The results reported in Figure~\ref{fig:CE} are averaged over 10 random tests. The convergence rate parameter $\gamma$ is set to $0.5$ for both algorithms. To generate the corruptions, we use $\alpha=0.1$ as the sparsity parameter, and the magnitude is controlled by $c=1$. Figure~\ref{fig:CE} confirms the efficiency of our algorithm. The left subfigure suggests that both ASAP and SAP have computational complexities that are linear with respect to the signal dimension, while ASAP has a much smaller constant in the front, as evident by the less steep slope in the plot. In the middle subfigure, we find that ASAP maintains its speed advantage under different ranks, and the advantage is more prominent for smaller ranks. The right subfigure provides empirical evidence for the linear convergence of ASAP, and also shows that ASAP can achieve over $10\times$ speedup when the problem size is large and the rank is small.

\begin{figure}[t]
\centering
\subfloat{\includegraphics[width=.333\linewidth]{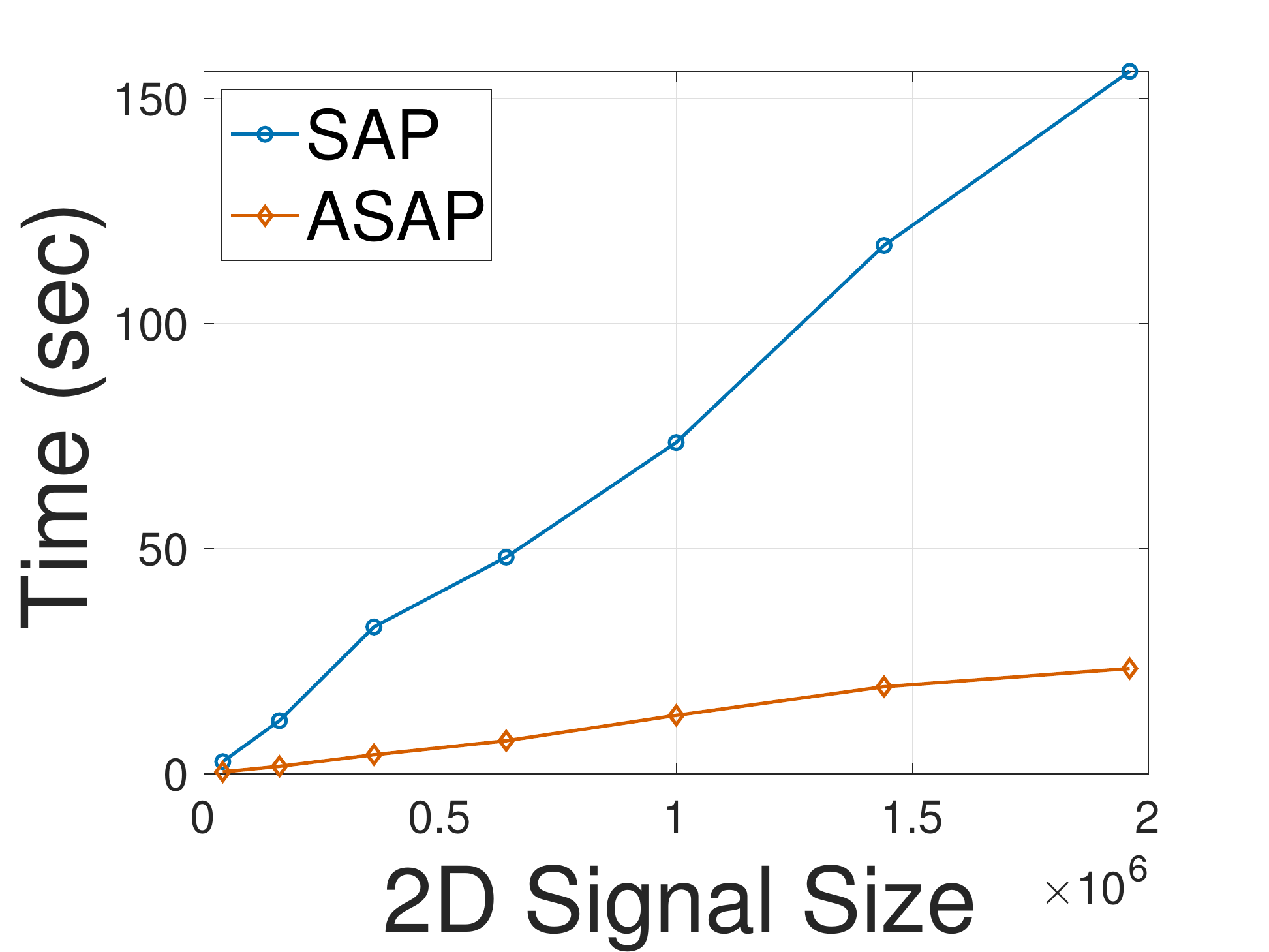}} \hfill
\subfloat{\includegraphics[width=.333\linewidth]{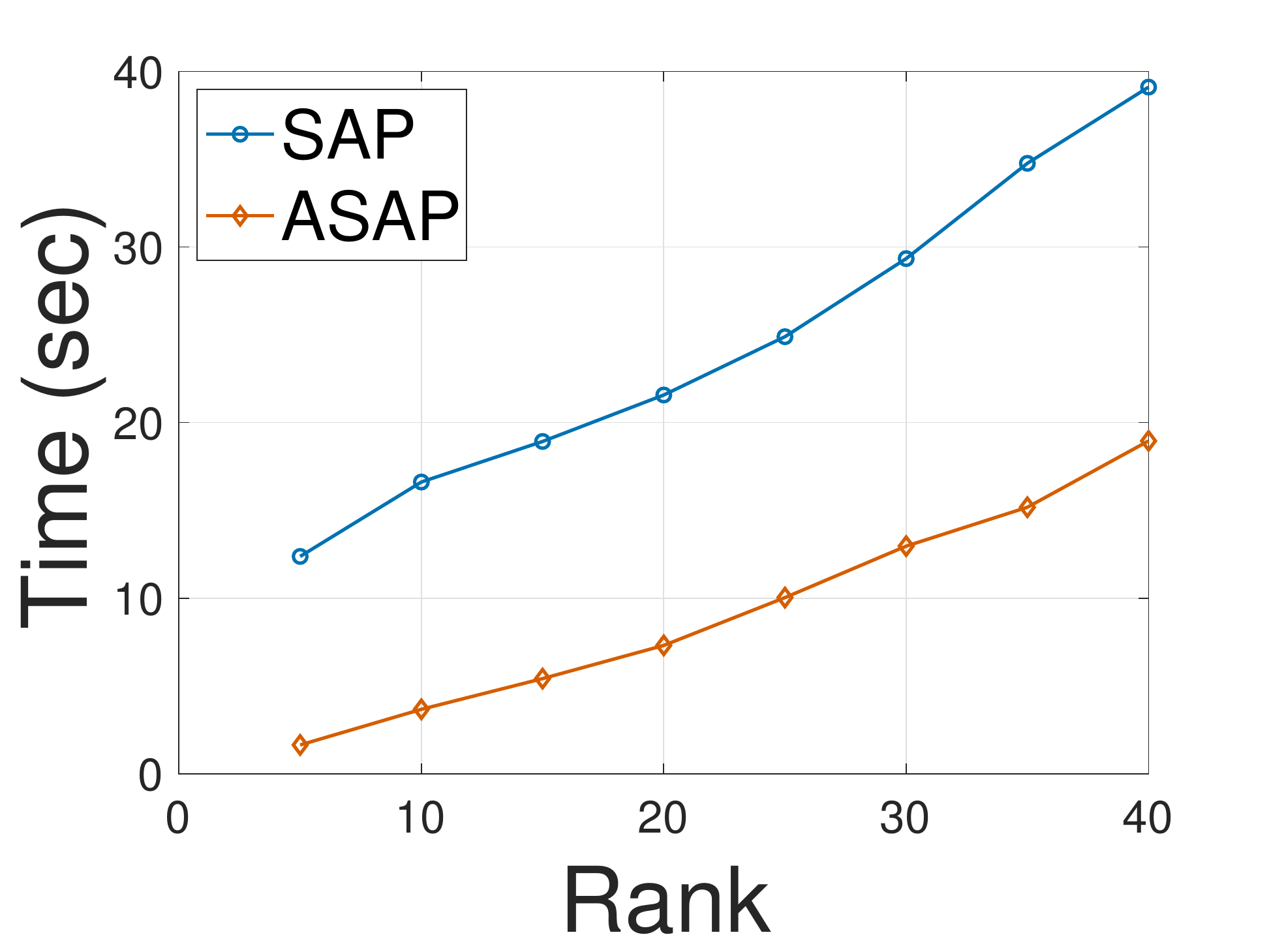}} \hfill
\subfloat{\includegraphics[width=.333\linewidth]{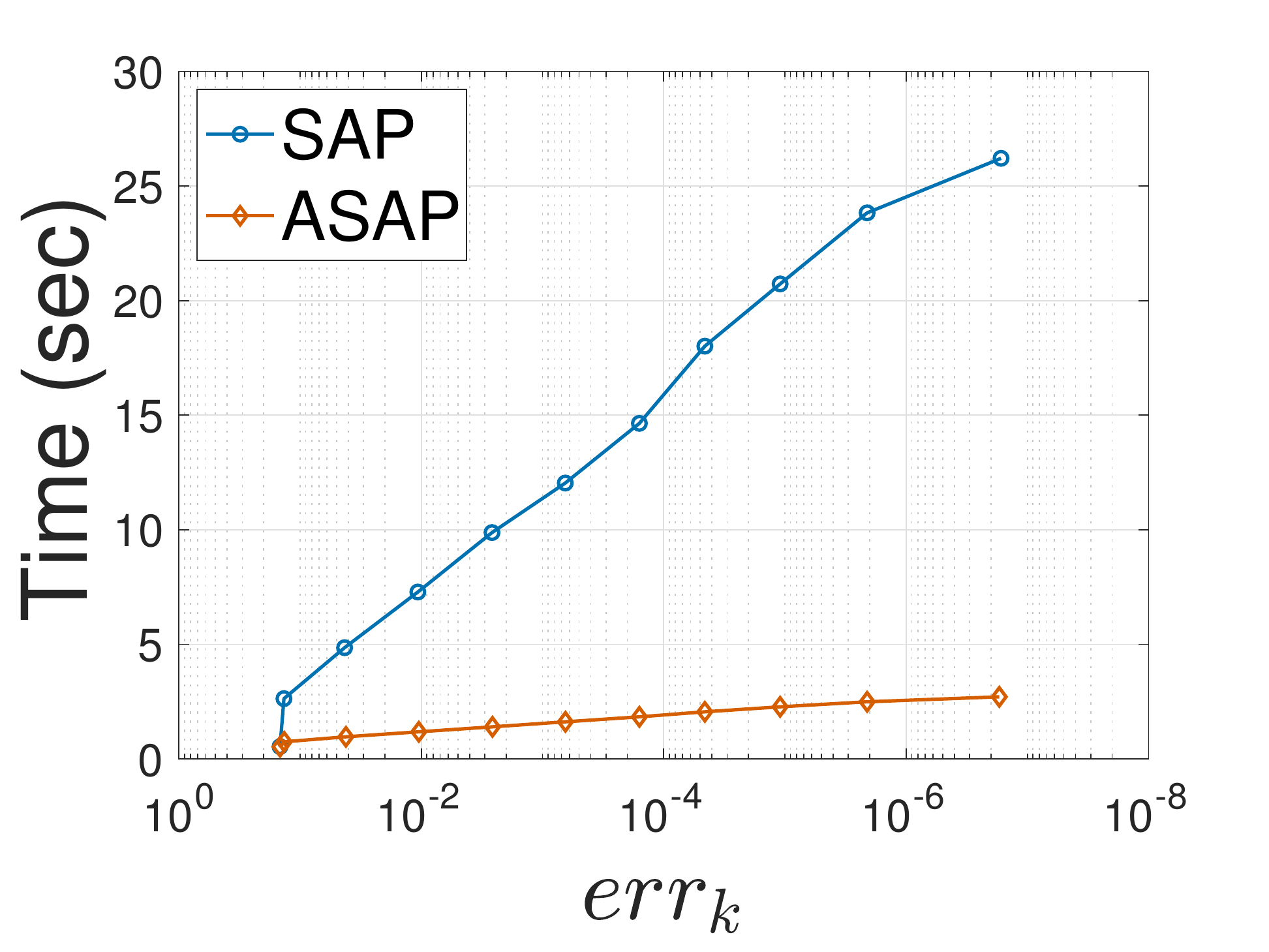}}
\caption{Computational efficiency comparisons. \textbf{Left:} fix rank $r=5$, runtime plots for 2D signals of size $200^2$, $400^2$, $600^2$, $800^2$, $1000^2$, $1200^2$, $1400^2$. \textbf{Middle:} fix 2D signals size to be $400^2$, runtime plots for rank $r=5,10,15,20,25,30,35,40$. \textbf{Right:} fix 2D signals size to be $400^2$ and $r=5$, relative error versus runtime.}\label{fig:CE}
\end{figure}

\subsection{Robustness to Additive Noise}

In practical applications, the observed measurements are usually corrupted by both additive noise and outliers, i.e.,
$$
\bm{z} = \bm{x}+\bm{s}+\bm{\eta},  
$$
where $\bm{\eta}$ is the noise term. We then test the robustness of ASAP to additive noise in the presence of outliers. For each test instance, we first generate a 2D signal $\bm{x}$ as in the previous subsection, then add i.i.d. Gaussian noise $\bm{\eta}$ such that $\bm{x}+\bm{\eta}$ is of certain SNR, and add the outliers at last. We experiment with SNR values from $80$ to $0$, and different amount of outliers ($\alpha\in\{0.1,0.3\}$) of different magnitudes ($c\in\{0.25,1,4\}$). The results shown in Figure~\ref{fig:additive_noise} are averaged over 10 random tests. For outliers of different magnitudes, we observe similar results. We also find that ASAP is robust to noise in the presence of different amount of outliers. Even for input signal corrupted by heavy noise (SNR$=0$), ASAP can still achieve very good recovery (SNR$>30$). Theoretically justifying such exceptional denoising ability is an interesting topic for further study.

\begin{figure}[t]
\centering
\subfloat{\includegraphics[width=.333\linewidth]{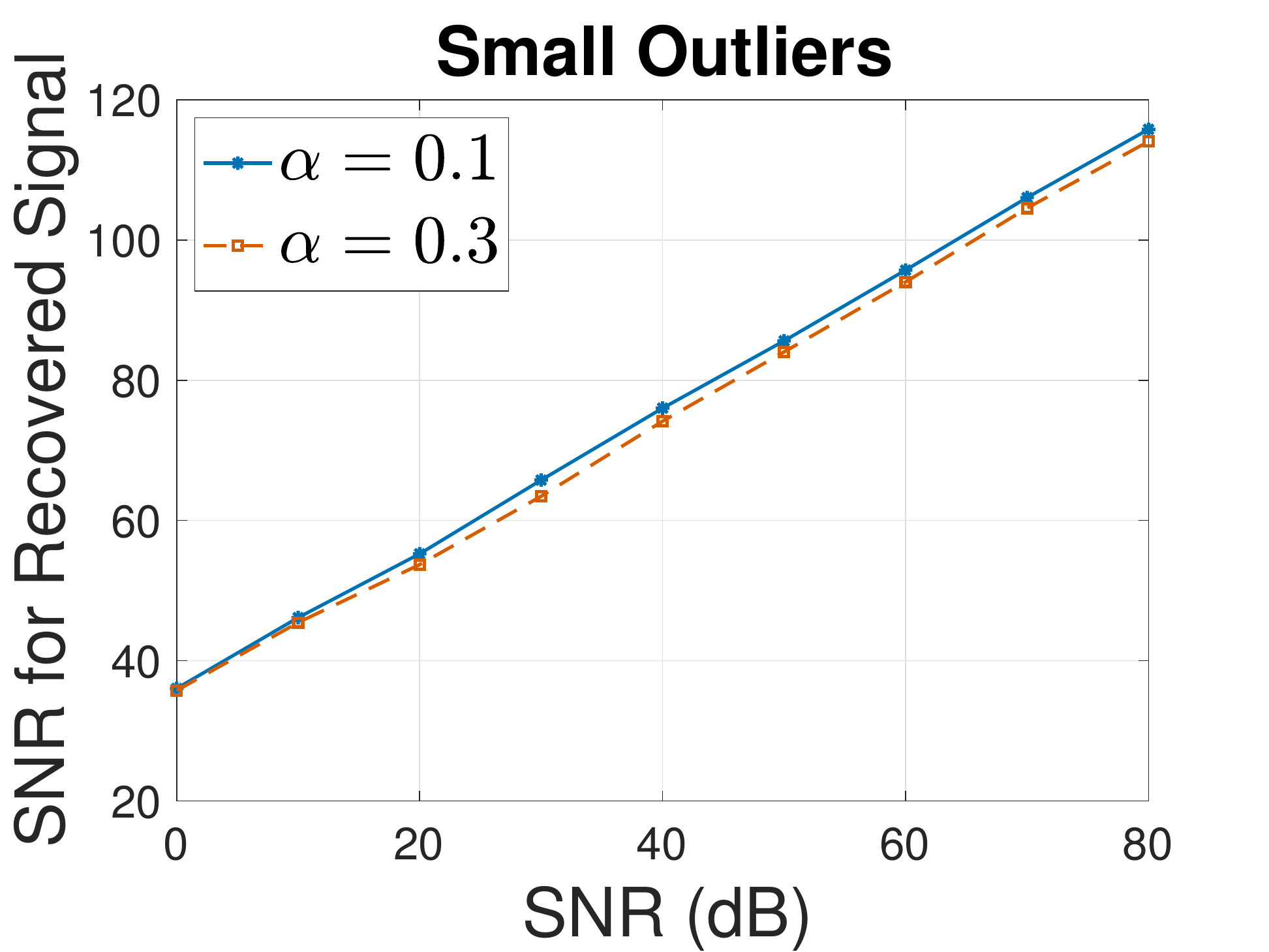}} \hfill
\subfloat{\includegraphics[width=.333\linewidth]{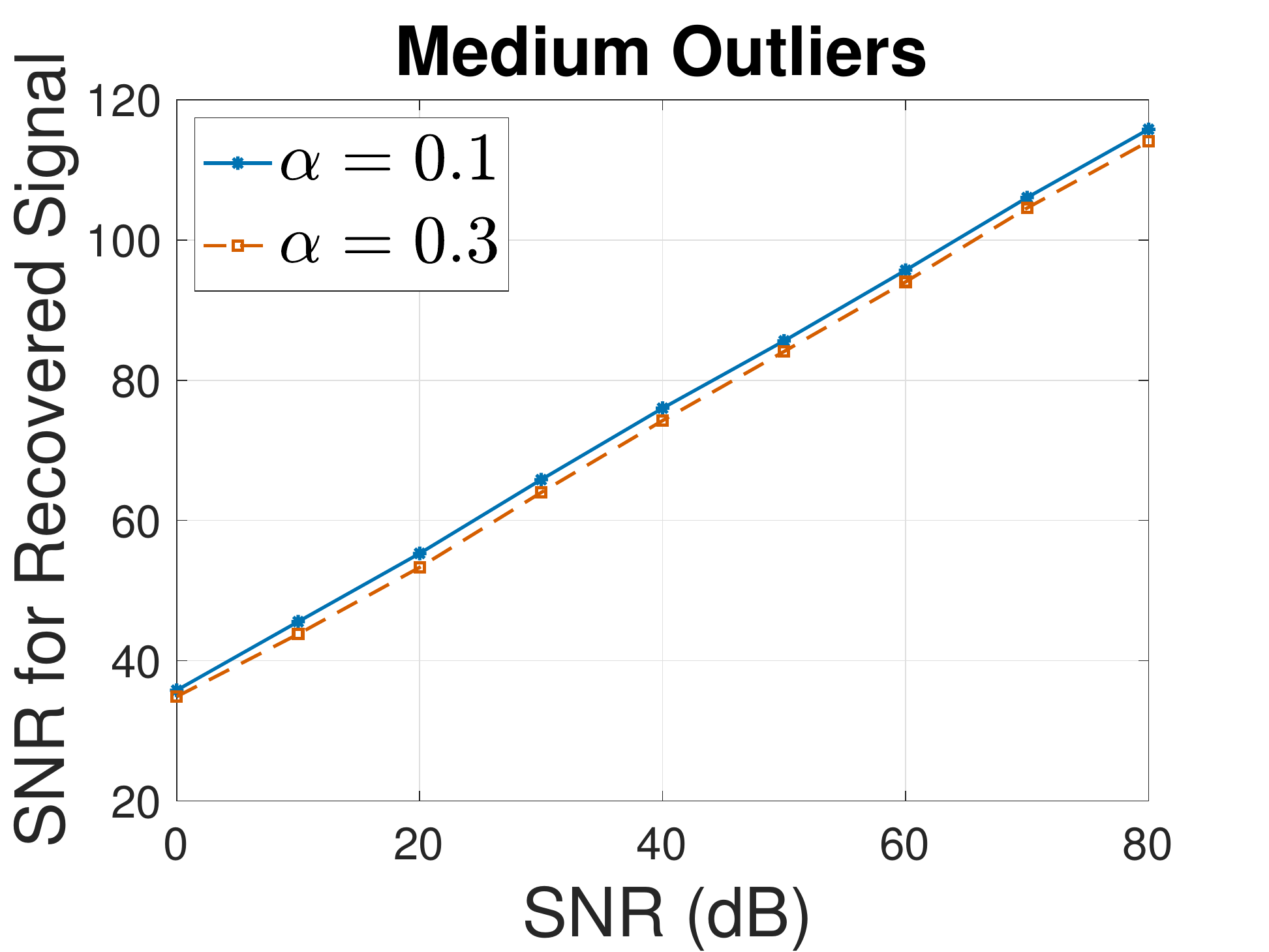}} \hfill
\subfloat{\includegraphics[width=.333\linewidth]{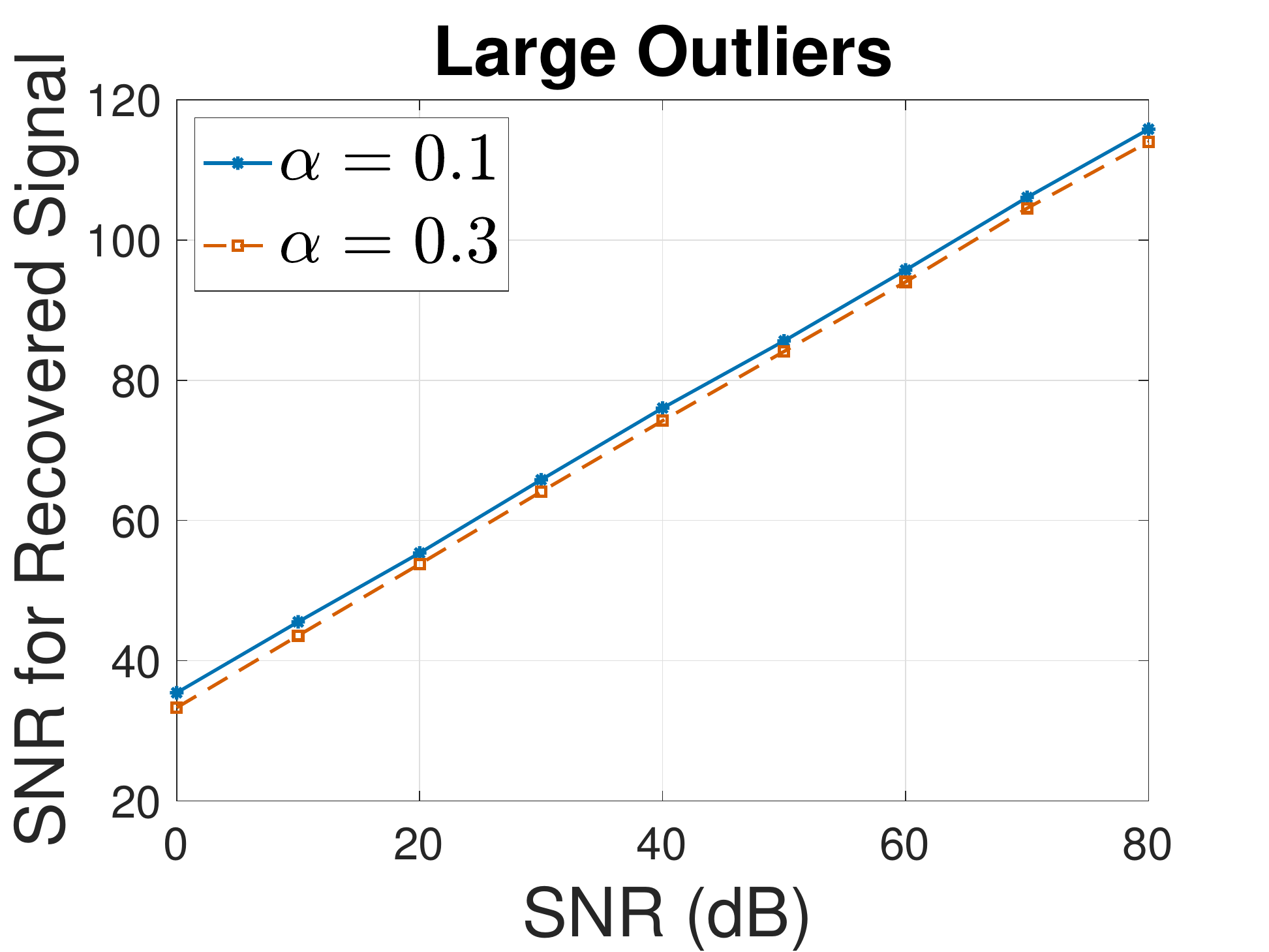}}
\caption{Robustness to additive noise in the presence of outliers. \textbf{Left:} small outliers ($c=0.25$). \textbf{Middle:} medium outliers ($c=1$). \textbf{Right:} large outliers ($c=4$).} \label{fig:additive_noise}
\end{figure}

\subsection{Impulse Corruptions in Nuclear Magnetic Resonance}

Our algorithm is applicable to removing impulse corruptions in signals from NMR spectroscopy \citep{qu2015accelerated}. The real-world data we are using is a 1D NMR signal of length $32,768$. Again, such size is  prohibited for the convex method Robust-EMaC. In this experiment, we add different amount of sparse outliers with large magnitude ($c=10$) to simulate the impulse corruptions caused by malfunctioning sensors, and compare our recovery result to SAP. For different values of $(\alpha,\gamma)$, we compare the computation time of SAP and ASAP. We can see from Table~\ref{table:nmr} that ASAP is more efficient in all the $4$ settings.

\begin{table}[t]
\caption{Computational time comparisons on the NMR data with different values of $(\alpha,\gamma)$.}\label{table:nmr}
\begin{center}
\begin{tabular}{ccccc}
\hline
\multicolumn{1}{|c|}{$(\alpha,\gamma)$} & \multicolumn{1}{c|}{$(0.2,0.6)$} & \multicolumn{1}{c|}{$(0.3,0.7)$} & \multicolumn{1}{c|}{$(0.4,0.8)$} & \multicolumn{1}{c|}{$(0.5,0.9)$} \\
\hline
\multicolumn{1}{|c|}{SAP} & \multicolumn{1}{c|}{$35.96$ s} & \multicolumn{1}{c|}{$55.21$ s} & \multicolumn{1}{c|}{$78.86$ s} & \multicolumn{1}{c|}{$150.19$ s} \\
\hline
\multicolumn{1}{|c|}{ASAP} & \multicolumn{1}{c|}{$11.51$ s} & \multicolumn{1}{c|}{$15.97$ s} & \multicolumn{1}{c|}{$19.10$ s} & \multicolumn{1}{c|}{$28.06$ s} \\
\hline
\end{tabular}
\end{center}
\end{table}

The two methods at convergence produce similar results, so we just show a typical recovery result for ASAP when $\alpha=0.5$ and $\gamma=0.9$. In Figure~\ref{fig:nmr}, we compare the power spectrum in a selected region of, the corrupted signal, the original noisy signal, and the result after corruption removal by ASAP (reversed for better visualization). We can see that the spectral peaks of the original noisy signal are well-preserved while the spurious peaks caused by the corruptions are removed. 

\begin{figure}[t]
\centering
\includegraphics[width=0.6\textwidth]{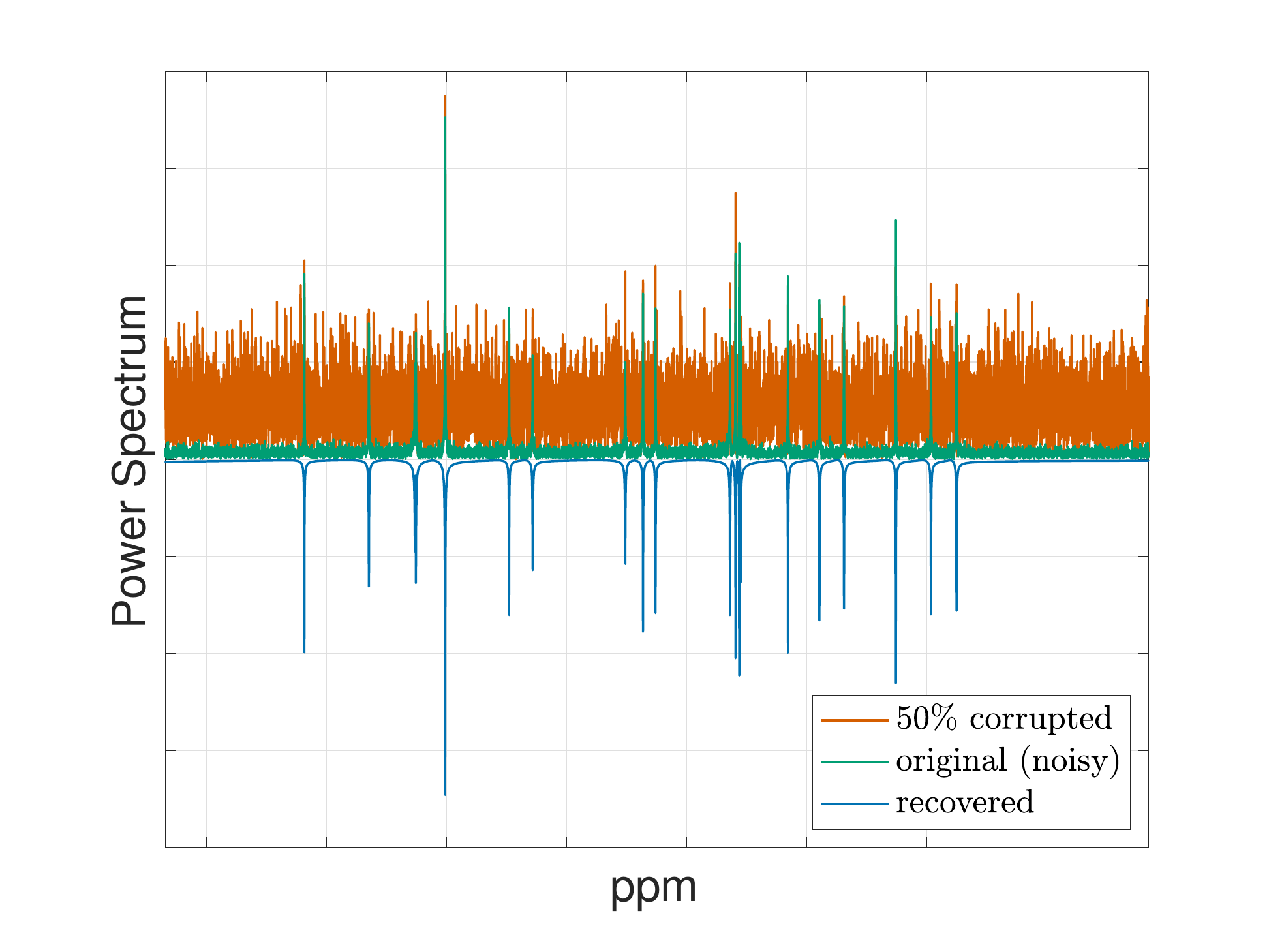}
\caption{Recovery result of ASAP on the NMR data.}\label{fig:nmr}
\end{figure}
\section{Proofs}  \label{sec:proof}

In this section, we present the mathematical proofs for the theoretical results in Theorem~\ref{thm:local convergence} and Theorem~\ref{thm:initialization bound}. 
Although we generally follow the framework in \citep{cai2017accelerated}, the details of the proofs are substantially different. For example, we improve the analysis by removing the unnecessary trimming step in \citep{cai2017accelerated}.

\subsection{Definitions and Auxiliary lemmas}

We first define some additional notations for the ease of presentation. Then some useful auxiliary lemmas are presented.

\begin{definition}\label{defs}
For any vector $\bz\in\C^n$, define an augmented Hermitian Hankel mapping $\widehat{\cH}$ as
    \begin{equation*}
        \widehat{\cH}(\bz)=
        \left[\begin{array}{cc}
        \bO & \cH(\bz) \\
        \left(\cH(\bz)\right)^{*} & \bO
        \end{array}\right] \in\C^{(n+1)\times (n+1)},
    \end{equation*}
where $\cH$ is defined as in \eqref{eq:Hankel definition}. For any matrix $\bM\in\C^{n_1\times n_2}$, its augmentation is defined using hat symbol, i.e.,
\begin{align*}
\widehat{\bM}=
    \left[\begin{array}{cc}
        \bO & \bM \\
        \bM^{*} & \bO
        \end{array}\right].
\end{align*}
For subspace projection $\mathcal{P}_{T_k}$, we also define its augmentation, acting on augmented Hermitian matrix $\widehat{\bM}$, as 
\begin{equation}  \label{eq:projections are same}
    \widehat{\mathcal{P}}_{T_k}(\widehat{\bM})= 
        \left[\begin{array}{cc}
        \bO & \mathcal{P}_{T_k}(\bM) \\
        \left(\mathcal{P}_{T_k}(\bM)\right)^{*} & \bO
        \end{array}\right].
\end{equation}
\end{definition}

\begin{lemma}  \label{lemma:augemtned property 1}
For any $\bz\in\C^n$, we have $\|\widehat{\cH}(\bz)\|_{\infty}= \|\cH(\bz)\|_{\infty} = \|\bz\|_{\infty}$. Also, for any $\bM\in\C^{n_1\times n_2}$,
we have $\|\cH^{\dagger}(\bM)\|_{\infty}\leq \|\bM\|_{\infty}$, where $\cH^{\dagger}$ is defined as in \eqref{eq:Hankel inverse mapping definition}. 
\end{lemma}
\begin{proof}
The results directly follow from the definitions of $\widehat{\cH}$, $\cH$ and $\cH^{\dagger}$.
\end{proof}

\begin{lemma}\label{lemma:augmentation}
For any $\bM\in\C^{n_1\times n_2}$,
consider the Hermitian matrix $\oM\in\C^{(n+1)\times (n+1)}$ augmented from $\bM$,
$$
\oM=\left[\begin{array}{cc}
\bO & \bM \\
\bM^{*} & \bO
\end{array}\right].
$$
Then
(1) $\|\oM\|_2=\|\bM\|_2$; (2) Suppose the SVD of $\bM$ can be written as $\bP\bm{\Delta}\bQ^*+\ddot{\bP}\ddot{\bm{\Delta}}\ddot{\bQ}^*$, where $\bP\bm{\Delta}\bQ^*$ is the best rank $r$ approximation of $\bM$. Then $\widehat{\bM}$ has SVD in the form of $\aug{\bP}\aug{\bm{\Delta}}\aug{\bQ}^*+\aug{\ddot{\bP}}\aug{\ddot{\bm{\Delta}}}\aug{\ddot{\bQ}}^*$, where 
\[
\aug{\bP}\aug{\bm{\Delta}}\aug{\bQ}^* =
\left[\begin{array}{cc}
\bO & \bP\bm{\Delta}\bQ^* \\
\bQ\bm{\Delta}\bP^* & \bO
\end{array}\right]
\]
is the best rank $2r$ approximation of $\oM$; (3) For $\mu$-incoherent rank $r$ matrix $\bL$, its augmentation has SVD $\widehat{\bL}= \aug{\bU}\aug{\bm{\Sigma}}\aug{\bV}^*$, $\aug{\bU}\aug{\bU}^*=\aug{\bV}\aug{\bV}^*$, and satisfy
$$
\|\aug{\bm{U}}\|_{2,\infty}\leq \sqrt{\frac{\mu c_s r}{n}}, \quad\textnormal{and}\quad  \|\aug{\bm{V}}\|_{2,\infty}\leq \sqrt{\frac{\mu c_s r}{n}}.
$$
\end{lemma}
\begin{proof}
Denote 
$$
\bm{R}:=\frac{1}{\sqrt{2}}\left[\begin{array}{cccc}
\bP & \ddot{\bP} & -\bP & -\ddot{\bP} \\
\bQ & \ddot{\bQ} & \bQ  & \ddot{\bQ}
\end{array}\right].
$$
It can be verified that
$$
\bm{R}
\left[\begin{array}{cccc}
\bm{\Delta} &  &  & \\
& \ddot{\bm{\Delta}} & & \\
& & -\bm{\Delta} & \\
& & & -\ddot{\bm{\Delta}}
\end{array}\right]
\bm{R}^*
$$
is an eigen-decomposition of $\oM$. Property (1) is then verified.
The best rank $2r$ approximation of $\oM$ can be written as
\begin{align*}
&\left(\frac{1}{\sqrt{2}}\left[\begin{array}{cc}
\bP & -\bP \\
\bQ & \bQ
\end{array}\right]\right)
\left[\begin{array}{cc}
\bm{\Delta} & \\
& -\bm{\Delta} 
\end{array}\right]
\left(\frac{1}{\sqrt{2}}\left[\begin{array}{cc}
\bP & -\bP \\
\bQ & \bQ
\end{array}\right]\right)^* \cr
&=
\left[\begin{array}{cc}
\bO & \bP\bm{\Delta}\bQ^* \\
\bQ\bm{\Delta}\bP^* & \bO
\end{array}\right].
\end{align*}
Property (3) then follows directly.
\end{proof}

\begin{lemma}  \label{lemma:bound of sparse matrix}
Let $\bs\in\C^n$ satisfy Assumption \nameref{assume:Sparse}. Then,
\[
\|\widehat{\cH}(\bs)\|_2 =\|\cH(\bs)\|_2 \leq \alpha n\|\bs\|_\infty.
\]
\end{lemma}
\begin{proof}
The Hermitian matrix $\widehat{\cH}(\bs)$ has no more than $\alpha n$ nonzero entries per row and column since $\bs$ is $\alpha$-sparse. Then 
\begin{align*}
    \|\cH(\bs)\|_2 = \|\ocH(\bs)\|_2 \leq \alpha n \|\ocH(\bs)\|_{\infty} = \alpha n\|\bs\|_\infty,
\end{align*}
where the first equality uses Lemma~\ref{lemma:augmentation}, and the inequality applies \citep[Lemma 4]{netrapalli2014non} for general Hermitian matrices with at most $\alpha n$ nonzeros entries per row and column. 
\end{proof}

\begin{lemma}  \label{lemma:error_of_projected_max_norm}
Let $\bs\in\C^n$ satisfy Assumption \nameref{assume:Sparse}. Let $\bL_k\in\mathbb{C}^{n_1\times n_2}$ be a  rank $r$ matrix with $4\mu\kappa^2$-incoherence. That is,
\begin{equation*}
\| \bU_k\|_{2,\infty}\leq  \sqrt{\frac{4\mu c_s r\kappa^2}{n}} \quad\textnormal{and}\quad  \| \bV_k\|_{2,\infty} \leq \sqrt{\frac{4\mu c_s r\kappa^2}{n}},
\end{equation*}
where $\bU_k\bm{\Sigma}_k\bV_k^\ast$ is SVD of $\bL_k$. 
If $\supp(\bs_{k+1})\subseteq\supp(\bs)$, then
\begin{equation*}
\|\mathcal{P}_{T_k} \cH(\bs-\bs_{k+1})\|_\infty \leq 12\alpha \mu c_s r \kappa^2 \|\bs-\bs_{k+1}\|_\infty .
\end{equation*}
\end{lemma}
\begin{proof}
Denote $\Omega:=\supp(\bs)$. By the incoherence assumption of $\bL_{k}$ and the sparsity of $\bs-\bs_{k+1}$, we have 
\begin{align*}
    &~[\mathcal{P}_{T_k}\cH(\bs-\bs_{k+1})]_{a,b} \cr
    =&~\langle  \mathcal{P}_{T_k}\cH(\bs-\bs_{k+1}),\bm{e}_a\bm{e}_b^T \rangle 
    ~=\langle \cH(\bs-\bs_{k+1}),\mathcal{P}_{T_k}(\bm{e}_a\bm{e}_b^T) \rangle  \cr
    =&~\langle \cH(\bs-\bs_{k+1}),\bU_k\bU_k^*\bm{e}_a\bm{e}_b^T   + \bm{e}_a\bm{e}_b^T\bV_k\bV_k^*   -  \bU_k\bU_k^*\bm{e}_a\bm{e}_b^T \bV_k\bV_k^* \rangle  \cr
    =&~  \langle \cH(\bs-\bs_{k+1})\bm{e}_b, \bU_k\bU_k^* \bm{e}_a\rangle + \langle \bm{e}_a^T\cH(\bs-\bs_{k+1}),\bm{e}_b^T \bV_k\bV_k^*\rangle  \\
    &~- \langle \cH(\bs-\bs_{k+1}),\bU_k\bU_k^*\bm{e}_a\bm{e}_b^T \bV_k\bV_k^* \rangle  \cr
    \leq&~ \Big(\sum_{i|(i,b)\in\Omega} |\bm{e}_i^T \bU_k\bU_k^*\bm{e}_a| + \sum_{j|(a,j)\in\Omega} |\bm{e}_b^T \bV_k\bV_k^*\bm{e}_j|\Big)\cdot\\
    &~\|\bs-\bs_{k+1}\|_\infty +  \|\cH(\bs-\bs_{k+1})\|_2~\|\bU_k\bU_k^*\bm{e}_a\bm{e}_b^T \bV_k\bV_k^*\|_\ast       \cr
    \leq&~ 2\alpha n\frac{4\mu c_s r \kappa^2}{n}\|\bs-\bs_{k+1}\|_\infty \cr
    &~ + \alpha n \|\bs-\bs_{k+1}\|_\infty\|\bU_k\bU_k^*\bm{e}_a\bm{e}_b^T \bV_k\bV_k^*\|_F \cr
                           \leq&~ 8\alpha \mu c_s r \kappa^2\|\bs-\bs_{k+1}\|_\infty  +  \alpha n \|\bs-\bs_{k+1}\|_\infty\frac{4\mu c_s r \kappa^2}{n} \cr
                           \leq&~ 12\alpha \mu c_s r \kappa^2 \|\bs-\bs_{k+1}\|_\infty,
\end{align*}
where the first inequality uses H\"older's inequality and the second inequality uses Lemma~\ref{lemma:bound of sparse matrix}. We also use the fact $\bU_k\bU_k^*\bm{e}_a\bm{e}_b^T\bV_k\bV_k^*$ is rank $1$ to bound its nuclear norm.
\end{proof}

\begin{lemma}   \label{lemma:P_T X 2-norm ineq}
Let $T$ be the direct sum of the row space and column space of $\bM$. Then, for any $\bZ$, we have
\[
\|\mathcal{P}_T \bZ\|_2 \leq \sqrt{4/3}\|\bZ\|_2.
\]
\end{lemma}
\begin{proof}
For $\bZ$ being Hermitian, this is an extension to the symmetric setting proved in \cite[Lemma~8]{cai2017accelerated}. For the non-Hermitian $\bZ$, consider the augmented Hermitian matrix $\oZ$.
\[
\oZ = \left[\begin{array}{cc}
\bO & \bZ \\
\bZ^* & \bO
\end{array}\right].
\]
Let $\bP\bm{\Delta}\bQ^*$ be the SVD of $\bM$. Denote
\[
\aug{\bm{P}} = \frac{1}{\sqrt{2}} \left[\begin{array}{cc}
\bP & -\bP \\
\bQ & \bQ 
\end{array}\right],
\]
$\aug{\bm{P}}$ is orthogonal and
\begin{align*}
\aug{\bm{P}}\aug{\bm{P}}^* \oZ + \oZ\aug{\bm{P}}\aug{\bm{P}}^* - \aug{\bm{P}}\aug{\bm{P}}^*\oZ\aug{\bm{P}}\aug{\bm{P}}^* = 
\left[\begin{array}{cc}
    \bO & \mathcal{P}_T \bZ  \\
    \left(\mathcal{P}_T \bZ\right)^* & \bO 
\end{array}\right].
\end{align*}
Apply Lemma~\ref{lemma:augmentation} and the result from Hermitian case, we have
\begin{align*}
\|\mathcal{P}_T \bZ\|_2 &= \left\|\left[\begin{array}{cc}
    \bO & \mathcal{P}_T \bZ  \\
    \left(\mathcal{P}_T \bZ\right)^* & \bO 
\end{array}\right]\right\|_2 \cr
&\leq \sqrt{4/3}\|\oZ\|_2~=\sqrt{4/3}\|\bZ\|_2
\end{align*}
for any $\bZ$.
\end{proof}

\begin{lemma} \label{init:lemma:bound_power_vector_norm_with_incoherence}
Let $\bs\in\C^n$ satisfy Assumption \nameref{assume:Sparse}, and $\bm{W}\in\mathbb{C}^{(n+1)\times r}$ be an orthogonal matrix with $\mu$-incoherence, i.e., $\|\bm{W}\|_{2,\infty}\leq\sqrt{\frac{\mu c_s r}{n}}$. Then,
\[
\|(\hatcH(\bs))^a\bm{W}\|_{2,\infty}\leq \sqrt{\frac{\mu c_s r}{n}}(\alpha n\|\bs\|_\infty)^a
\]
for all integer $a\geq 0$.
\end{lemma}
\begin{proof}
This proof is also done by mathematical induction. When $a=0$, $\|\bm{W}\|_{2,\infty}\leq\sqrt{\frac{\mu c_sr}{n}}$ is satisfied following the assumption. So we have the base case. Now, we assume $\|(\hatcH(\bs))^a\bm{W}\|_{2,\infty}\leq \sqrt{\frac{\mu c_s r}{n}}(\alpha n\|\bs\|_\infty)^a $. Then, we can show
\begin{align*}
\|\hatcH(\bs)^{a+1}\bm{W}\|_{2,\infty} 
&= \|\hatcH(\bs)\hatcH(\bs)^a\bm{W}\|_{2,\infty} \cr
&\leq \|\hatcH(\bs)\|_{1,\infty}\|\hatcH(\bs)^a\bm{W}\|_{2,\infty} \cr
&\leq \alpha n\|\hatcH(\bs)\|_{\infty}\sqrt{\frac{\mu c_s r}{n}}(\alpha n\|\bs\|_\infty)^a \cr
&= \sqrt{\frac{\mu c_s r}{n}}(\alpha n\|\bs\|_\infty)^{a+1}
\end{align*}
where the first inequality uses $\|\bm{A}\bm{B}\|_{2,\infty}\leq \|\bm{A}\|_{1,\infty}\|\bm{B}\|_{2,\infty}$ and the second inequality uses Assumption~\nameref{assume:Sparse} and the induction hypothesis. This finishes the proof.
\end{proof}

\subsection{Local Analysis}

We proceed via a series of lemmas characterizing, in some neighborhood of the ground truth, the behaviors of the iterates produced by ASAP, and establish Theorem~\ref{thm:local convergence} at the end of this subsection. For the ease of notations, we use $\tau:=4\alpha\mu c_s r\kappa$ and $\upsilon:=3\tau(2\sqrt{\mu c_s r}\kappa+\mu c_s r \kappa)$ in the subsequent proofs. Under Assumption~\nameref{assume:Sparse}, $\tau\lesssim\mathcal{O}(1/\mu c_s r\kappa )$, and $\upsilon\lesssim\mathcal{O}(1)$. Also recall that $\bL=\cH(\bx)$, $\sigma_i^{x}$ denotes the $i$-th singular value of $\cH(\bx)$, and $\bL_k$ is the estimate of $\bL$ at the $k$-th iteration.

\begin{lemma} \label{lemma:Bound_eigenvalues}
Let $\sigma^{(k)}_{i}$ denote the $i$-th singular value of $\bL_k$.  
If \[
\|\bL-\bL_k\|_2 \leq 2\tau \gamma^k\sigma_r^x
\]
for some $\gamma \in(0,1)$, then 
\begin{equation}  \label{eq:Bound_eigenvalues 2}
(1-2\tau) \sigma_i^x\leq  \sigma^{(k)}_i \leq(1+2\tau) \sigma_i^x
\end{equation}
hold for $1\leq i\leq r$ and $k\geq 0$.
\end{lemma}
\begin{proof}
By Weyl's inequality and the fact $\gamma<1$, we have
\begin{equation}  \label{eq:Bound_eigenvalues 1}
|\sigma_i^x-\sigma^{(k)}_i| \leq\|\bL-\bL_k\|_2\leq 2\tau\sigma_r^x,
\end{equation}
which implies the claim immediately, since both $\bL$ and $\bL_k$ are rank $r$ matrices.
\end{proof}

\begin{lemma}   \label{lemma:Bound_of_S-S_k}
Let $\bs\in\C^n$ satisfy Assumption \nameref{assume:Sparse}.
Recall that  $\beta=\frac{\mu c_s r}{2\kappa n}$. If 
\[
\|\bL-\bL_k\|_2 \leq 2\tau \gamma^k\sigma_r^x,\quad
\|\bx-\bx_k\|_\infty \leq \frac{\tau-2\tau^2}{8\alpha \kappa n} \gamma^k\sigma_r^x,
\]
then we have
\[
\supp(\bs_{k+1})\subseteq \supp(\bs), ~ \|\bs-\bs_{k+1}\|_\infty \leq  \frac{\tau}{4\alpha \kappa n } \gamma^k\sigma_r^x.
\]
\end{lemma}
\begin{proof}
Denote $\Omega := \supp(\bs)$. Notice that, for every entry of $\bs_{k+1}$, we have
\begin{align*}
[\bs_{k+1}]_{t}&=[\mathcal{T}_{\zeta_{k+1}} (\bz-\bx_{k})]_{t}=[\mathcal{T}_{\zeta_{k+1}} (\bs+\bx-\bx_{k})]_{t} \cr
&=
\begin{cases}
\mathcal{T}_{\zeta_{k+1}} ([\bs+\bx-\bx_{k}]_{t}) &\quad t\in\Omega \cr
\mathcal{T}_{\zeta_{k+1}} ([\bx-\bx_{k}]_{t})        &\quad t\in\Omega^c  \cr
\end{cases}.
\end{align*}
Let $\sigma_{i}^{(k)}$ denote the $i$-th singular value of $\bL_k$. By Lemma~\ref{lemma:Bound_eigenvalues},
\begin{align*}
\|\bx-\bx_{k}\|_\infty
&\leq  \frac{\tau(1-2\tau)}{8\alpha\kappa n } \gamma^k\sigma_r^x \\
&\leq \frac{\tau(1-2\tau)}{8\alpha \kappa^2 n} \frac{1}{1-2\tau} \gamma^{k}\sigma_1^{(k)} 
= \zeta_{k+1},
\end{align*}
which implies $[\bs_{k+1}]_{t}=0$ for all $t\in \Omega^c$; in other words, $\supp(\bs_{k+1})\subseteq \Omega$. Denote $\Omega_{k+1}:=\supp(\bs_{k+1})$. Then, for every entry of $\bs-\bs_{k+1}$ in $\Omega$, it holds that 
\begin{align*}
[\bs-\bs_{k+1}]_{t} &=
\begin{cases}
[\bx_{k}-\bx]_{t} &   \cr
[\bs]_{t}         & \cr
\end{cases} \leq
\begin{cases}
\|\bx-\bx_{k}\|_\infty       &  \cr
\|\bx-\bx_{k}\|_\infty +\zeta_{k+1} & \cr
\end{cases}\cr
&\leq
\begin{cases}
\frac{\tau(1-2\tau)}{8\alpha \kappa n } \gamma^k\sigma_r^x       & (i,j)\in \Omega_{k+1}   \cr
\frac{\tau}{4\alpha\kappa n } \gamma^k\sigma_r^x & (i,j)\in \Omega\backslash\Omega_{k+1}. \cr
\end{cases}
\end{align*}
The last step holds since $\zeta_{k+1}\leq \frac{\tau(1+2\tau)}{8\alpha\kappa n } \gamma^k\sigma_r^x$, which follows from Lemma~\ref{lemma:Bound_eigenvalues}. 
\end{proof}

\begin{lemma} \label{lemma:norm_of_Z}
Let $\bx,\bs\in\C^n$ satisfy Assumptions \nameref{assume:Inco} and \nameref{assume:Sparse}, respectively. If $\bL_k$ is $4\mu\kappa^2$-incoherent and
\[
\|\bL-\bL_k\|_2 \leq 2\tau \gamma^k\sigma_r^x,\quad
\|\bx-\bx_k\|_\infty \leq \frac{\tau-2\tau^2}{8\alpha\kappa  n} \gamma^k\sigma_r^x,
\]
then we have
\begin{align}  
\|(\mathcal{P}_{T_k}-\mathcal{I})\bL+\mathcal{P}_{T_k}\cH(\bs-\bs_{k+1})\|_2 \leq \tau\gamma^{k+1}\sigma_r^x,\label{eq:norm_of_Z 1} \\
\|(\mathcal{P}_{T_k}-\mathcal{I})\bL+\mathcal{P}_{T_k}\cH(\bs-\bs_{k+1})\|_{2,\infty} \leq n^{-\frac12} \upsilon\gamma^{k}\sigma_r^x,\label{eq:norm_of_Z 2}
\end{align}
provided $1>\gamma\geq 4\tau + \sqrt{1/12}$.
\end{lemma}
\begin{proof}
From \citep[Lemma~6]{cai2017accelerated}, we obtain the inequality
\begin{align*}
    \|(\mathcal{P}_{T_k}-\mathcal{I})\bL\|_2 &\leq \|\bL-\bL_k\|_2^2/\sigma_r^x.
\end{align*}
First note that Lemma~\ref{lemma:Bound_of_S-S_k} holds. We further get
\begin{align*}
&\quad~\|(\mathcal{P}_{T_k}-\mathcal{I})\bL+\mathcal{P}_{T_k}\cH(\bs-\bs_{k+1})\|_2 \cr
&\leq \|\bL-\bL_k\|_2^2/\sigma_r^x + \sqrt{4/3}\|\cH(\bs-\bs_{k+1})\|_2  \cr
&\leq 2\tau\gamma^k\|\bL-\bL_k\|_2+ \sqrt{4/3}\alpha n\|\bs-\bs_{k+1}\|_\infty  \cr
&\leq \left( 4\tau + \sqrt{1/12}\right) \tau\gamma^k\sigma_r^x 
~\leq \tau\gamma^{k+1}\sigma_r^x,
\end{align*}
where the first inequality uses Lemma~\ref{lemma:P_T X 2-norm ineq}, the second inequality uses Lemma~\ref{lemma:bound of sparse matrix}, and the last inequality uses the bound of $\gamma$. 

To bound $\|(\mathcal{P}_{T_k}-\mathcal{I})\bL+\mathcal{P}_{T_k}\cH(\bs-\bs_k)\|_{2,\infty}$, note that
\begin{align*}
&\quad~ \|(\mathcal{I}-\mathcal{P}_{T_k})\bL\|_{2,\infty}\cr &= \max_l \|\bm{e}_l^T(\bU\bU^*-\bU_k\bU^*_k)(\bL-\bL_k)(\bm{I}-\bV_k\bV^*_k)\|_2  \cr
&\leq (2\kappa + 1 )\sqrt{\frac{\mu c_s r}{n}}\|\bL-\bL_k\|_2 
~\leq 3\kappa \sqrt{\frac{\mu c_s r}{n}}\|\bL-\bL_k\|_2,
\end{align*}
where the first inequality holds since $\bL$ is $\mu$-incoherent and $\bL_k$ is $4\mu \kappa^2$-incoherent.
Hence, we have
\begin{align*}
&~\|(\mathcal{P}_{T_k}-\mathcal{I})\bL+\mathcal{P}_{T_k}\cH(\bs-\bs_{k+1})\|_{2,\infty} \cr
    \leq&~ \|(\mathcal{I}-\mathcal{P}_{T_k})\bL\|_{2,\infty}+\|\mathcal{P}_{T_k}\cH(\bs-\bs_{k+1})\|_{2,\infty} \cr
	\leq&~ 3\kappa\sqrt{\frac{\mu c_s r}{n}}\|\bL-\bL_k\|_2 + \sqrt{n}\|\mathcal{P}_{T_k}\cH(\bs-\bs_{k+1})\|_\infty\cr
	\leq&~ 3 \kappa\sqrt{\frac{\mu c_s r}{n}}\|\bL-\bL_k\|_2 +12\alpha \mu c_s r \kappa^2 \sqrt{n}\|\bs-\bs_{k+1}\|_\infty\cr
	\leq&~ n^{-1/2}\upsilon\gamma^{k}\sigma_r^x,
\end{align*}
where the third inequality uses Lemma~\ref{lemma:error_of_projected_max_norm}, and the last inequality follows from Lemma~\ref{lemma:Bound_of_S-S_k} and the definition of $\upsilon$.
\end{proof}

\begin{lemma} \label{lemma:Bound_of_L-L_k_2_norm}
Let $\bx,\bs\in\C^n$ satisfy Assumptions \nameref{assume:Inco} and \nameref{assume:Sparse}, respectively. If $\bL_k$ is $4\mu\kappa^2$-incoherent and
\[
\|\bL-\bL_k\|_2 \leq 2\tau \gamma^k\sigma_r^x,\quad
\|\bx-\bx_k\|_\infty \leq \frac{\tau-2\tau^2}{8\alpha  \kappa n} \gamma^k\sigma_r^x,
\]
then we have 
\[
\|\bL-\bL_{k+1}\|_2 \leq 2\tau \gamma^{k+1}\sigma_r^x,\quad
\]
provided $1>\gamma\geq 4\tau + \sqrt{1/12}$.
\end{lemma}
\begin{proof}
A direct calculation yields 
\begin{align*}
&\quad~\|\bL-\bL_{k+1}\|_2 \cr &\leq\|\bL-\mathcal{P}_{T_k}\cH(\bz-\bs_{k+1})\|_2+\|\mathcal{P}_{T_k}\cH(\bz-\bs_{k+1})-\bL_{k+1}\|_2\cr
&\leq 2\|\bL-\mathcal{P}_{T_k}\cH(\bz-\bs_{k+1})\|_2\cr
&=2\|\bL-\mathcal{P}_{T_k}(\bL+\cH(\bs-\bs_{k+1}))\|_2\cr
&= 2\|(\mathcal{P}_{T_k}-\mathcal{I})\bL+\mathcal{P}_{T_k}\cH(\bs-\bs_{k+1})\|_2 \leq 2 \tau\gamma^{k+1}\sigma_r^x, 
\end{align*}
where the second inequality holds since $\bL_{k+1}$ is the best rank $r$ approximation to $\mathcal{P}_{{T_k}}\cH(\bz-\bs_{k+1})$, and the last inequality follows from \eqref{eq:norm_of_Z 1}. 
\end{proof}

\begin{lemma}  \label{lemma:Bound_of_L-L_k_infty_norm}
Let $\bx,\bs\in\C^n$ satisfy Assumptions \nameref{assume:Inco} and \nameref{assume:Sparse}, respectively.
If $\bL_k$ is $4\mu\kappa^2$-incoherent and
\[
\|\bL-\bL_k\|_2 \leq 2\tau \gamma^k\sigma_r^x,\quad
\|\bx-\bx_k\|_\infty \leq \frac{\tau-2\tau^2}{8\alpha  \kappa n} \gamma^k\sigma_r^x,
\]
then we have 
\[
\|\bx-\bx_{k+1}\|_\infty \leq \frac{\tau-2\tau^2}{8\alpha  \kappa n} \gamma^{k+1}\sigma_r^x,
\]
provided $1>\gamma\geq\max\{4\tau +\sqrt{1/12},\frac{4\upsilon}{(1-12\tau)(1-\tau-\upsilon)^2}\}$.
\end{lemma}
\begin{proof}
From Lemma~\ref{lemma:augemtned property 1}, $\|\bx-\bx_{k+1}\|_\infty \leq \|\bL-\bL_{k+1}\|_\infty$. We then provide a bound for $\|\bL-\bL_{k+1}\|_\infty$.

Consider the augmented matrix of $\mathcal{P}_{T_k}\cH(\bz-\bs_{k+1})$: 
\begin{align*}
    \widehat{\mathcal{P}}_{T_k}\hatcH(\bz-\bs_{k+1})
    &=\left[\begin{array}{cc}
    \bO & \mathcal{P}_{T_k}\cH(\bz-\bs_{k+1}) \\
    \mathcal{P}_{T_k}\cH(\bz-\bs_{k+1})^* & \bO 
    \end{array}\right] \cr
    &:= \aug{\bm{U}}_{k+1}\bLa\aug{\bm{U}}_{k+1}^*+\aug{\ddot{\bU}}_{k+1}\ddot{\bLa}\aug{\ddot{\bU}}_{k+1}^*,
\end{align*}
where the eigen-decomposition is derived from the SVD of $\mathcal{P}_{T_k}\cH(\bz-\bs_{k+1})$ and $\aug{\bm{U}}_{k+1}\bLa\aug{\bm{U}}_{k+1}^*$ is the best rank $2r$ approximation, as shown in Lemma~\ref{lemma:augmentation}. Moreover,
\begin{align*}
    \aug{\bm{U}}_{k+1}\bLa\aug{\bm{U}}_{k+1}^* = \left[\begin{array}{cc}
    \bO & \bL_{k+1} \\
    \bL_{k+1}^* & \bO 
    \end{array}\right].
\end{align*}
The $i$-th singular value of $\mathcal{P}_{T_k}\cH(\bz-\bs_{k+1})$ is denoted by $\sigma_i$ to ease the notations in proving this lemma.  

Denote $\hatZ=\widehat{\mathcal{P}}_{T_k}\hatcH(\bz-\bs_{k+1})-\hatL=(\widehat{\mathcal{P}}_{T_k}-\mathcal{I})\widehat{\bL}+\widehat{\mathcal{P}}_{T_k}\hatcH(\bs-\bs_{k+1})$. For $1\leq j\leq 2r$, let $\bm{u}_j$ be the $j$-th eigenvector of $\aug{\bm{U}}_{k+1}$. Noticing that $(\oL+\oZ)\bm{u}_j=\lambda_j\bm{u}_j$, we have 
\begin{align*}
\bm{u}_j = \lb\bm{I}-\frac{\oZ}{\lambda_j}\rb^{-1}  \frac{\oL}{\lambda_j}\bm{u}_j= \sum_{l=0}^\infty \left(\frac{\oZ}{\lambda_j}\right)^l\frac{\oL}{\lambda_j}\bm{u}_j
\end{align*}
for all $\bm{u}_j$, where the expansion is valid because for $1\leq i\leq r$,
$$\frac{\|\oZ\|_2}{|\lambda_{i}|}=\frac{\|\oZ\|_2}{|\lambda_{i+r}|}=\frac{\ln\bZ\rn_2}{\sigma_{i}}\leq\frac{\ln\bZ\rn_2}{\sigma_{r}}\leq\frac{\tau}{1-\tau}<1,$$
where in the second last inequality, we use
\eqref{eq:norm_of_Z 1} in Lemma~\ref{lemma:norm_of_Z} and Lemma~\ref{lemma:Bound_eigenvalues}. Then
\begin{align*}
&\quad~\aug{\bm{U}}_{k+1}\bLa\aug{\bm{U}}_{k+1}^* = \sum_{j=1}^{2r} \bm{u}_j\lambda_j\bm{u}_j^*   \cr
&= \sum_{j=1}^{2r} \Bigg( \sum_{a\geq0} \left(\frac{\oZ}{\lambda_j}\right)^a\frac{\oL}{\lambda_j}\Bigg)\bm{u}_j\lambda_j\bm{u}_j^* \Bigg( \sum_{b\geq0} \left(\frac{\oZ}{\lambda_j}\right)^b\frac{\oL}{\lambda_j} \Bigg)^* \cr
&= \sum_{a\geq0}  (\oZ)^a \oL \sum_{j=1}^{2r}\left(\bm{u}_j\frac{1}{\lambda_j^{a+b+1}}\bm{u}_j^* \right) \oL \sum_{b\geq0} (\oZ)^b  \cr
&= \sum_{a,b\geq0}  (\oZ)^a\oL\aug{\bm{U}}_{k+1}\bm{\Lambda} ^{-(a+b+1)}\aug{\bm{U}}_{k+1}^*\oL(\oZ)^b,
\end{align*}
and therefore
\begin{align*}
\hatL_{k+1}-\hatL 
=&\aug{\bm{U}}_{k+1}\bLa\aug{\bm{U}}_{k+1}^* -\oL \cr
= &\oL\aug{\bm{U}}_{k+1}\bm{\Lambda} ^{-1}\aug{\bm{U}}_{k+1}^*\oL-\oL  \cr
& + \sum_{a+b\geq1}  (\oZ)^a\oL\aug{\bm{U}}_{k+1}\bm{\Lambda} ^{-(a+b+1)}\aug{\bm{U}}_{k+1}^*\oL(\oZ)^b  \cr
:=& \bm{Y}_0 + \sum_{a+b\geq1} \bm{Y}_{ab}.
\end{align*}
Hence,
\begin{align*}
\|\hatL_{k+1}-\hatL\|_\infty &\leq \|\bm{Y}_0\|_\infty + \sum_{a+b\geq1} \|\bm{Y}_{ab}\|_\infty.
\end{align*}
We will handle $\bm{Y}_0$ first.  
\begin{align*}
\|\bm{Y}_0\|_\infty &= \max_{ij} |\bm{e}_i^T(\oL\aug{\bm{U}}_{k+1}\bm{\Lambda} ^{-1}\aug{\bm{U}}_{k+1}^*\oL-\oL)\bm{e}_j| \cr
&=\max_{ij} |\bm{e}_i^T\aug{\bm{U}}\aug{\bm{U}}^*(\oL\aug{\bm{U}}_{k+1}\bm{\Lambda} ^{-1}\aug{\bm{U}}_{k+1}^*\oL-\oL)\aug{\bm{U}}\aug{\bm{U}}^*\bm{e}_j| \cr
& \leq \|\aug{\bm{U}}\|_{2,\infty}\|\oL\aug{\bm{U}}_{k+1}\bm{\Lambda} ^{-1}\aug{\bm{U}}_{k+1}^*\oL-\oL\|_2 \|\aug{\bm{U}}\|_{2,\infty} \cr
&\leq \frac{\mu c_s r}{n}\|\oL\aug{\bm{U}}_{k+1}\bm{\Lambda} ^{-1}\aug{\bm{U}}_{k+1}^*\oL-\oL\|_2,
\end{align*}
where the second equality is due to the fact $\oL=\aug{\bm{U}}\aug{\bm{U}}^*\oL=\oL\aug{\bm{U}}\aug{\bm{U}}^*$. 
Since $\oL=\aug{\bm{U}}_{k+1}\bm{\Lambda}\aug{\bm{U}}_{k+1}^*+\aug{\ddot{\bU}}_{k+1}\ddot{\bm{\Lambda}}\aug{\ddot{\bU}}_{k+1}^* -\oZ$,
\begin{align*}
&~\|\oL\aug{\bm{U}}_{k+1}\bm{\Lambda} ^{-1}\aug{\bm{U}}_{k+1}^*\oL-\oL\|_2 \cr
=&~ \|(\aug{\bm{U}}_{k+1}\bm{\Lambda}\aug{\bm{U}}_{k+1}^*+\aug{\ddot{\bU}}_{k+1}\ddot{\bm{\Lambda}}\aug{\ddot{\bU}}_{k+1}^* -\oZ)\aug{\bm{U}}_{k+1}\bm{\Lambda} ^{-1}\aug{\bm{U}}_{k+1}^* \cr
&~~(\aug{\bm{U}}_{k+1}\bm{\Lambda}\aug{\bm{U}}_{k+1}^*+\aug{\ddot{\bU}}_{k+1}\ddot{\bm{\Lambda}}\aug{\ddot{\bU}}_{k+1}^* -\oZ)-\oL\|_2  \cr
=&~\|\aug{\bm{U}}_{k+1}\bm{\Lambda}\aug{\bm{U}}_{k+1}^*-\oL-\aug{\bm{U}}_{k+1}\aug{\bm{U}}_{k+1}^*\oZ-\oZ\aug{\bm{U}}_{k+1}\aug{\bm{U}}_{k+1}^*\cr
&~~+\oZ\aug{\bm{U}}_{k+1}\bm{\Lambda} ^{-1}\aug{\bm{U}}_{k+1}^*\oZ\|_2  \cr
\leq&~ \|\oZ-\aug{\ddot{\bU}}_{k+1}\ddot{\bm{\Lambda}}\aug{\ddot{\bU}}_{k+1}^*\|_2 +2\|\oZ\|_2+\frac{\|\oZ\|_2^2}{\sigma_r}  \cr
\leq&~ \|\aug{\ddot{\bU}}_{k+1}\ddot{\bm{\Lambda}}\aug{\ddot{\bU}}_{k+1}^*\|_2 +3\|\oZ\|_2+\frac{\|\bZ\|_2^2}{\sigma_r}  \cr
\leq&~\|\aug{\ddot{\bU}}_{k+1}\ddot{\bm{\Lambda}}\aug{\ddot{\bU}}_{k+1}^*\|_2 +3\|\bZ\|_2+\|\bZ\|_2  \cr
\leq&~  \sigma_{r+1}+4\|\bZ\|_2 ~\leq  5\|\bZ\|_2,
\end{align*}
where the third inequality holds since $\frac{\|\bZ\|_2}{\sigma_r}\leq\frac{\tau}{1-\tau}<1$, and the last inequality follows from $\sigma_{r+1}=\sigma_{r+1}-\sigma_{r+1}^x\leq\|\bZ\|_2$ since $\cH(\bx)$ is a rank $r$ matrix. Together, by \eqref{eq:norm_of_Z 1} in Lemma~\ref{lemma:norm_of_Z},
\begin{align}  
\|\bm{Y}_0\|_\infty
&\leq \frac{\mu c_sr}{n}5\tau\gamma^{k+1}\sigma_r^x. \label{eq:Y0 bound}
\end{align}
Next, we derive the bound for $\bm{Y}_{ab}$. Note that
\begin{align*}
&\quad~ \|\bm{Y}_{ab}\|_\infty \cr
&=\max_{ij} |\bm{e}_i^T\oZ^a\oL\aug{\bm{U}}_{k+1}\bm{\Lambda} ^{-(a+b+1)}\aug{\bm{U}}_{k+1}^*\oL\oZ^b\bm{e}_j|  \cr
&=\max_{ij} |\bm{e}_i^T\oZ^a\aug{\bm{U}}\aug{\bm{U}}^*\oL\aug{\bm{U}}_{k+1}\bm{\Lambda} ^{-(a+b+1)}\aug{\bm{U}}_{k+1}^*\oL\aug{\bm{U}}\aug{\bm{U}}^*\oZ^b\bm{e}_j|  \cr
&\leq  \|\oZ^a\aug{\bm{U}}\|_{2,\infty}\|\oL\aug{\bm{U}}_{k+1}\bm{\Lambda} ^{-(a+b+1)}\aug{\bm{U}}_{k+1}^*\oL\|_2\|\oZ^b\aug{\bm{U}}\|_{2,\infty}  \cr
&\leq \frac{\mu c_sr}{n}( \sqrt{n}\|\oZ\|_{2,\infty})^{a+b} \|\oL\aug{\bm{U}}_{k+1}\bm{\Lambda} ^{-(a+b+1)}\aug{\bm{U}}_{k+1}^*\oL\|_2,
\end{align*}
where the last inequality follows from \citep[Lemma~9]{cai2017accelerated}, which states that
\[
\|\hatZ^a\aug{\bm{U}}\|_{2,\infty}\leq \sqrt{\frac{\mu c_sr}{n}}(\sqrt{n}\|\hatZ\|_{2,\infty})^a
\]
holds for all $a\geq 0$. On the other hand, 
\begin{align*}
&~\|\oL\aug{\bm{U}}_{k+1}\bm{\Lambda} ^{-(a+b+1)}\aug{\bm{U}}_{k+1}^*\oL\|_2 \cr
=&~ \|\aug{\bm{U}}_{k+1}\bm{\Lambda} ^{-(a+b-1)}\aug{\bm{U}}_{k+1}^*-\aug{\bm{U}}_{k+1}\bm{\Lambda} ^{-(a+b)}\aug{\bm{U}}_{k+1}^*\oZ\cr
 &~-\oZ\aug{\bm{U}}_{k+1}\bm{\Lambda} ^{-(a+b)}\aug{\bm{U}}_{k+1}^*
+\oZ\aug{\bm{U}}_{k+1}\bm{\Lambda} ^{-(a+b+1)}\aug{\bm{U}}_{k+1}^*\oZ\|_2 \cr
\leq&~ \sigma_r^{-(a+b-1)} + 2\sigma_r^{-(a+b)}\|\oZ\|_2+ \sigma_r^{-(a+b+1)}\|\oZ\|_2^2 \cr
=&~ \sigma_r^{-(a+b-1)}\left( 1+ \frac{2\|\bZ\|_2}{\sigma_r}+\frac{\|\bZ\|_2^2}{\sigma_r^2} \right)  \cr
=&~ \sigma_r^{-(a+b-1)}\left( 1+ \frac{\|\bZ\|_2}{\sigma_r} \right)^2  \cr
\leq&~ \sigma_r^{-(a+b-1)}\left(  \frac{1}{1-\tau} \right)^2
\leq\left( \frac{1}{1-\tau} \right)^2 \left((1-\tau)\sigma_r^x\right)^{-(a+b-1)},
\end{align*}
where the second inequality follows from $\frac{\|\bZ\|_2}{\sigma_r}\leq\frac{\tau}{1-\tau}$, and the last inequality follows from Lemma~\ref{lemma:Bound_eigenvalues}. Together with \eqref{eq:norm_of_Z 2} in Lemma~\ref{lemma:norm_of_Z}, we have
\begin{align*}
\sum_{a+b\geq1}\|\bm{Y}_{ab}\|_\infty 
&\leq \frac{\mu c_sr}{n}\left( \frac{1}{1-\tau} \right)^2 \upsilon \gamma^{k}\sigma_r^x  \sum_{a+b\geq1} \left(  \frac{\upsilon}{1-\tau}\right)^{a+b-1}\cr
&\leq \frac{\mu c_sr}{n}\left( \frac{1}{1-\tau} \right)^2 \upsilon \gamma^{k}\sigma_r^x  \cdot 2\left( \frac{1}{1-\frac{\upsilon}{1-\tau}}\right)^2\cr
&=\frac{\mu c_sr}{n} \frac{2\upsilon}{(1-\tau-\upsilon)^2} \gamma^{k}\sigma_r^x,
\end{align*}
where $\upsilon<1-\tau$ is satisfied when the constant hidden in the bound of $\alpha$ is small enough. Finally, combining with \eqref{eq:Y0 bound} gives
\begin{align*}
\|\bL-\bL_{k+1}\|_\infty &= \|\hatL-\hatL_{k+1}\|_\infty \leq \|\bm{Y}_0\|_\infty + \sum_{a+b\geq1} \|\bm{Y}_{ab}\|_\infty  \cr
&\leq \left(5\tau\gamma+ \frac{2\upsilon}{(1-\tau-\upsilon)^2}\right)\frac{\mu c_sr}{n} \gamma^{k}\sigma_r^x\cr
&\leq \frac{1-2\tau}{2}\frac{\mu c_sr}{n} \gamma^{k+1}\sigma_r^x =\frac{\tau(1-2\tau)}{8\alpha  \kappa n} \gamma^{k+1}\sigma_r^x,
\end{align*}
where the third inequality uses $\gamma\geq\frac{4\upsilon}{(1-12\tau)(1-\tau-\upsilon)^2}$ and the last step uses the definition of $\tau$.
\end{proof}

\begin{lemma}  \label{lemma:keep incoherence of L}
Let $\bx,\bs\in\C^n$ satisfy Assumptions \nameref{assume:Inco} and \nameref{assume:Sparse}, respectively.
If $\bL_k$ is $4\mu\kappa^2$-incoherent and
\[
\|\bL-\bL_k\|_2 \leq 2\tau \gamma^k\sigma_r^x,\quad
\|\bx-\bx_k\|_\infty \leq \frac{\tau-2\tau^2}{8\alpha  \kappa n} \gamma^k\sigma_r^x,
\]
then $\bL_{k+1}$ is also $4\mu \kappa^2$-incoherent.
\end{lemma}
\begin{proof}
Following the proof and notations of Lemma~\ref{lemma:Bound_of_L-L_k_infty_norm}, we can similarly show
\begin{align*}
\|\hatL_{k+1}-\hatL\|_{2,\infty} &\leq \|\bm{Y}_0\|_{2,\infty} + \sum_{a+b\geq1} \|\bm{Y}_{ab}\|_{2,\infty},
\end{align*}
where
\begin{align*}
\|\bm{Y}_0\|_{2,\infty}
&\leq \sqrt{\frac{\mu c_s r}{n}}\|\oL\aug{\bm{U}}_{k+1}\bm{\Lambda} ^{-1}\aug{\bm{U}}_{k+1}^*\oL-\oL\|_2 \cr
 &\leq \sqrt{\frac{\mu c_sr}{n}}5\tau\gamma^{k+1}\sigma_r^x, \numberthis\label{eq:Y0 2,inf bound}
\end{align*}
and
\begin{align*}
    &\quad~\|\bm{Y}_{ab}\|_{2,\infty}\cr
    &\leq \sqrt{\frac{\mu c_sr}{n}}( \sqrt{n}\|\oZ\|_{2,\infty})^{a+b} \|\oL\aug{\bm{U}}_{k+1}\bm{\Lambda} ^{-(a+b+1)}\aug{\bm{U}}_{k+1}^*\oL\|_2
\end{align*}
since $\|\hatZ\|_2\leq\sqrt{n}\|\hatZ\|_{2,\infty}$.
Furthermore, we can show
\begin{align*}
\sum_{a+b\geq1}\|\bm{Y}_{ab}\|_{2,\infty} 
    &\leq\sqrt{\frac{\mu c_sr}{n}} \frac{2\upsilon}{(1-\tau-\upsilon)^2}\gamma^{k}\sigma_r^x. 
\end{align*}
Combining with \eqref{eq:Y0 2,inf bound}, we get
\begin{align*}
    \|\bL-\bL_{k+1}\|_{2,\infty} \leq  \frac{1-2\tau}{2}\sqrt{\frac{\mu c_sr}{n}} \gamma^{k+1}\sigma_r^x\leq \frac{1}{2} \sqrt{\frac{\mu c_sr}{n}} \sigma_r^x.
\end{align*}
Since $\bL=\cH(\bx)$ is $\mu$-incoherent, $\|\bL\|_{2,\infty}\leq \sqrt{\frac{\mu c_sr}{n}} \sigma_1^x$, which implies
\begin{align*}
    \|\bL_{k+1}\|_{2,\infty} \leq  \frac{3}{2}\sqrt{\frac{\mu c_sr}{n}} \sigma_1^x.
\end{align*}
Let $\bU_{k+1}\bm{\Sigma}_{k+1}\bV_{k+1}^*$ be the SVD of $\bL_{k+1}$. We obtain
\begin{align*}
    \|\bU_{k+1}\|_{2,\infty}&=\|\bL_{k+1}\bV_{k+1}\bm{\Sigma}_{k+1}^{-1}\|_{2,\infty} \cr
    &\leq \frac{3}{2}\sqrt{\frac{\mu c_sr}{n}} \frac{\sigma_1^x}{\sigma_r}
    ~\leq 2\kappa\sqrt{\frac{\mu c_sr}{n}}, 
\end{align*}
where the last step uses lemmas~\ref{lemma:Bound_eigenvalues} and \ref{lemma:Bound_of_L-L_k_2_norm}, ensuring that $\sigma_r \geq \frac{3}{4}\sigma_r^x$. Similarly, we can also show
$
\|\bV_{k+1}\|_{2,\infty}
\leq 2\kappa\sqrt{\frac{\mu c_sr}{n}} 
$.
We conclude that $\bL_{k+1}$ is $4\mu\kappa^2$-incoherent.
\end{proof}

Now, we have all the ingredients for the proof of Theorem~\ref{thm:local convergence}, which shows the local linear convergence of Algorithm~\ref{Algo:Algo1}.
\begin{proof} [Proof of Theorem~\ref{thm:local convergence}] This theorem will be proved by  mathematical induction.\\
\textbf{Base Case:} When $k=0$, the base case is satisfied by the assumption on the initialization.\\
\textbf{Induction Step:} Assume that $\bL_k$ is $4\mu\kappa^2$-incoherent,
\[
\|\bL-\bL_k\|_2 \leq 2\tau \gamma^k\sigma_r^x,\quad
\|\bx-\bx_k\|_\infty \leq \frac{\tau-2\tau^2}{8\alpha  \kappa n} \gamma^k\sigma_r^x,
\]
at the $k$-th iteration. At the $(k+1)$-th iteration. It follows directly from lemmas~\ref{lemma:Bound_of_L-L_k_2_norm}, \ref{lemma:Bound_of_L-L_k_infty_norm} and \ref{lemma:keep incoherence of L} that $\bL_{k+1}$ is also $4\mu\kappa^2$-incoherent, and
\begin{align*}
\|\bL-\bL_{k+1}\|_2 \leq 2\tau \gamma^{k+1}\sigma_r^x, \end{align*}
\begin{align*}
\|\bx-\bx_{k+1}\|_\infty \leq \frac{\tau-2\tau^2}{8\alpha  \kappa n} \gamma^{k+1}\sigma_r^x,
\end{align*}
which completes the proof. Additionally, notice that we overall require
$1>\gamma\geq\max\{4\tau + \sqrt{1/12},\frac{4\upsilon}{(1-12\tau)(1-\tau-\upsilon)^2}\}$. 
By the definition of $\tau$ and $\upsilon$, one can easily see that the lower bound approaches $\sqrt{1/12}$ when the constant hidden in the bound of $\alpha$ is sufficiently small. Therefore,  the theorem can be proved for any $\gamma\in\left(\frac{1}{\sqrt{12}},1\right)$.
\end{proof}

\subsection{Initialization}

Finally, we show Algorithm~\ref{Algo:Init1} provides sufficient initialization for the local convergence. 

\begin{proof} [Proof of Theorem~\ref{thm:initialization bound}]
The proof is partitioned into several parts.\\
\noindent\textbf{Part 1:} Note that 
\[
\|\cH(\bx)\|_\infty 
\leq  \|\bU\|_{2,\infty}\|\bm{\Sigma}\|_2\|\bV\|_{2,\infty}\leq\frac{\mu c_s r}{n}\sigma_1^x,
\]
where the last inequality follows from the assumption that $\cH(\bx)$ is $\mu$-incoherent.
Thus, with the choice of $\beta_{init}\geq\frac{\mu c_s r\sigma_1^x}{n\sigma_1(\cH(\bz))}$, we have
\begin{equation*}
\|\bx\|_{\infty}=
\|\cH(\bx)\|_\infty\leq \beta_{init}\sigma_1(\cH(\bz)) = \zeta_{0}.
\end{equation*}
From here, following the proof in Lemma~\ref{lemma:Bound_of_S-S_k}, we can conclude \begin{align}
    &\supp(\bs_0)\subseteq\supp(\bs), \cr
    &\|\bs-\bs_{0}\|_\infty \leq \|\bx\|_\infty+\zeta_0 
    \leq\frac{4\mu c_s r}{n}\sigma_1^x,  \label{eq:SminusS}
\end{align}
where the last inequality follows from $\beta_{init}\leq\frac{3\mu c_s r\sigma_1^x}{n\sigma_1(\cH(\bz))}$, which implies $\zeta_{0}\leq \frac{3\mu c_s r}{n}\sigma_1^x$.

\noindent\textbf{Part 2:} To bound the approximation error of $\bL_0$  to $\bL=\cH(\bx)$ in terms of the spectral norm, note that
\begin{align*}
\|\bL-\bL_0\|_2 & = \|\cH(\bx)-\mathcal{D}_r\cH(\bz-\bs_{0})\|_2 \cr
        & \leq 2 \|\cH(\bx)-\cH(\bz-\bs_{0})\|_2 
        ~= 2 \|\cH(\bs-\bs_{0})\|_2,
\end{align*}
where the inequality holds since $\mathcal{D}_r\cH(\bz-\bs_{0})$ is the best rank $r$ approximation to $\cH(\bz-\bs_{0})$. By Lemma~\ref{lemma:bound of sparse matrix},
\begin{equation}\label{eq:norm:L-L0}
\|\bL-\bL_0\|_2\leq 8\alpha \mu c_s r\sigma_1^x=2\tau\sigma_r^x.
\end{equation}
This proves the first claim of Theorem~\ref{thm:initialization bound}.

\noindent\textbf{Part 3:}
For $\|\bx-\bx_{0}\|_\infty$, it is upper bounded by $\|\bL-\bL_0\|_{\infty}$. 
Note that $\cH(\bz-\bs_{0})=\bL+\cH(\bs-\bs_{0})$. Let $\sigma_i$ denote the $i$-th singular value of $\cH(\bz-\bs_{0})$. Applying  Weyl's inequality together with Lemma~\ref{lemma:bound of sparse matrix}, we get
\begin{equation}   \label{init:eq:eigenvalues bound 0}
|\sigma_i^x-\sigma_i| \leq \|\cH(\bs-\bs_{0})\|_2 \leq \alpha n\|\bs-\bs_{0}\|_\infty \leq\tau\sigma_r^x
\end{equation}
holds for all $i$. Consequently, we have 
\begin{align*} 
&(1-\tau)\sigma_i^x \leq \sigma_i \leq (1+\tau)\sigma_i^x,\qquad \forall 1\leq i\leq r.\numberthis  \label{init:eq:eigenvalues bound 1}
\end{align*}
Consider the augmented matrix of $\cH(\bz-\bs_{0})$,
\begin{align*}
    \hatcH(\bz-\bs_{0})&=\left[\begin{array}{cc}
    \bO & \cH(\bz-\bs_{0}) \\
    \cH(\bz-\bs_{0})^* & \bO 
    \end{array}\right] \cr
    &:= \aug{\bm{U}}_0\bLa\aug{\bm{U}}_0^*+\aug{\ddot{\bU}}_0\ddot{\bLa}\aug{\ddot{\bU}}_0^*,
\end{align*}
where the eigen-decomposition is derived as in Lemma~\ref{lemma:augmentation}.
Denote $\bZ=\cH(\bz-\bs_{0})-\bL=\cH(\bs-\bs_{0})$. Following the proof in Lemma~\ref{lemma:Bound_of_L-L_k_infty_norm}, we get
\begin{align*}
&\quad~\|\hatL_0-\hatL\|_\infty
= \|\aug{\bm{U}}_0\bm{\Lambda}\aug{\bm{U}}_0^* -\oL\|_\infty \cr
&\leq \|\oL\aug{\bm{U}}_0\bm{\Lambda}^{-1}\aug{\bm{U}}_0^*\oL-\oL\|_\infty \cr
&\quad~ + \sum_{a+b\geq1} \|(\oZ)^a\oL\aug{\bm{U}}_0\bm{\Lambda}^{-(a+b+1)}\aug{\bm{U}}_0^*\oL(\oZ)^b \|_\infty \cr
&:= \|\bm{Y}_0\|_\infty + \sum_{a+b\geq1} \|\bm{Y}_{ab}\|_\infty,
\end{align*}
For $\bm{Y}_0$,
\begin{equation}  \label{init:eq:Y0 bound}
\|\bm{Y}_0\|_\infty\leq \frac{5\mu c_s r}{n} \|\oZ\|_2\leq 5\alpha\mu c_sr \|\oZ\|_\infty,
\end{equation}
where the last inequality is due to Lemma~\ref{lemma:bound of sparse matrix}, and $\|\oZ\|_\infty=\|\bs-\bs_0\|_\infty$. For $\bm{Y}_{ab}$,
\begin{align*}
\|\bm{Y}_{ab}\|_\infty
      &\leq \|\oZ^a\aug{\bm{U}}\|_{2,\infty}\|\oL\aug{\bm{U}}_0\bm{\Lambda} ^{-(a+b+1)}\aug{\bm{U}}_0^*\oL\|_2\|\oZ^b\aug{\bm{U}}\|_{2,\infty}.
\end{align*}
By Lemma~\ref{init:lemma:bound_power_vector_norm_with_incoherence}, we have
\begin{align*}
\|\bm{Y}_{ab}\|_\infty&\leq \frac{\mu c_s r}{n}(\alpha n\|\oZ\|_\infty)^{a+b}\|\oL\aug{\bm{U}}_0\bm{\Lambda} ^{-(a+b+1)}\aug{\bm{U}}_0^*\oL\|_2.
\end{align*}
Similar to Lemma~\ref{lemma:Bound_of_L-L_k_infty_norm}, we can show 
\begin{align*}
    \|\oL\aug{\bm{U}}_0\bm{\Lambda} ^{-(a+b+1)}\aug{\bm{U}}_0^*\oL\|_2&\leq \frac65\sigma_r^{-(a+b-1)} \cr
    &\leq \frac65\left((1-\tau)\sigma_r^x\right)^{-(a+b-1)}.
\end{align*}
Hence, we have
\begin{align*}
\sum_{a+b\geq1}\|\bm{Y}_{ab}\|_\infty\
&\leq  \frac65\alpha\mu c_s r\|\oZ\|_\infty\sum_{a+b\geq1}\left(\frac{\alpha n\|\oZ\|_\infty}{(1-\tau)\sigma_r^x}\right)^{a+b-1} \cr
&\leq  \frac65\alpha\mu c_s r\|\oZ\|_\infty\sum_{a+b\geq1}\left(\frac{\tau}{1-\tau}\right)^{a+b-1} \cr
&\leq  3\alpha\mu c_s r\|\oZ\|_\infty.
\end{align*}
Finally, combining \eqref{init:eq:Y0 bound} and above yields
\begin{align*}
\|\bL_0-\bL\|_\infty &=\|\hatL_0-\hatL\|_\infty\leq \|\bm{Y}_0\|_\infty + \sum_{a+b\geq1} \|\bm{Y}_{ab}\|_\infty  \cr
                     &\leq 5\alpha \mu c_sr \|\oZ\|_\infty + 3\alpha \mu c_sr \|\oZ\|_\infty \cr
                     &=    8\alpha \mu c_sr \|\bs-\bs_0\|_\infty 
              \leq \frac{\tau-2\tau^2}{8\alpha  \kappa n} \sigma_r^x,
\end{align*}
where the last step uses \eqref{eq:SminusS}, 
and the bound of $\alpha$ in Assumption \nameref{assume:Sparse}. 
The second claim of Theorem~\ref{thm:initialization bound} is then proved.

\noindent\textbf{Part 4:} 
Following the proof and notation in part 3, and similar to Lemma~\ref{lemma:keep incoherence of L}, we can get
\begin{align*}
\|\bL_0-\bL\|_{2,\infty} &\leq \|\bm{Y}_0\|_{2,\infty} + \sum_{a+b\geq1} \|\bm{Y}_{ab}\|_{2,\infty}  \cr
                     &\leq 5\alpha \sqrt{\mu c_s r n}\|\oZ\|_\infty + 3\alpha \sqrt{\mu c_s r n}\|\oZ\|_\infty \cr
                     &=  8\alpha \sqrt{\mu c_sr n}\|\bs-\bs_0\|_\infty
             \leq \frac{1}{2}\sqrt{\frac{\mu c_sr}{n}}\sigma_r^x.
\end{align*}
This implies
\begin{align*}
    \|\bL_0\|_{2,\infty}\leq \|\bL\|_{2,\infty} + \|\bL_0-\bL\|_{2,\infty}
    \leq  \frac{3}{2}\sqrt{\frac{\mu c_sr}{n}}\sigma_1^x.
\end{align*}
Let $\bU_0\bm{\Sigma}_0\bV_0^*$ be the SVD of $\bL_0$. Hence, by \eqref{init:eq:eigenvalues bound 0}, we have
\begin{align*}
    \|\bU_0\|_{2,\infty}&=\|\bL_0\bV_0\bm{\Sigma}_0^{-1}\|_{2,\infty} \leq \frac{3}{2}\sqrt{\frac{\mu c_sr}{n}} \frac{\sigma_1^x}{\sigma_r}
    \leq 2\kappa\sqrt{\frac{\mu c_sr}{n}} .
\end{align*}
Similarly, we can also show
$
    \|\bV_0\|_{2,\infty}
    \leq 2\kappa\sqrt{\frac{\mu c_sr}{n}} 
$. 
Hence, we conclude $\bL_0$ is $4\mu\kappa^2$-incoherent.
\end{proof}
\section{Conclusion} \label{sec:conclusion}
In this paper, we propose a highly efficient non-convex algorithm, dubbed ASAP, to achieve the robust recovery of low-rank Hankel matrices, with application to corrupted spectrally sparse signals. Guaranteed exact recovery with a linear convergence rate has been established for ASAP. Numerical experiments, compared with convex and non-convex methods in the literature, confirm its computational efficiency and robustness to corruptions. The experiments also suggest the derived tolerance of corruptions is highly pessimistic, and one possible further direction is to improve the theoretical analysis. It would also be interesting to theoretically justify the algorithm's exceptional robustness to noise in the presence of outliers, as observed in the experiment section. Another further research direction is to extend the algorithm and analysis to the missing data case. 


\bibliographystyle{IEEEtran}
\bibliography{IEEEabrv,ref}

\end{document}